\documentclass[11pt]{amsart}
\usepackage{amsmath}
\usepackage{graphicx}
\usepackage{bm}
\usepackage{subfig,tikz}
\usepackage{wrapfig}
\usepackage{xspace}
\usepackage{enumerate}
\theoremstyle{plain}
\newtheorem{theorem}{Theorem}[section]
\newtheorem{lemma}[theorem]{Lemma}

\newtheorem{proposition}[theorem]{Proposition}

\theoremstyle{definition}

\numberwithin{equation}{section}

\begin{document}

\title[Area-Angular Momentum-Charge Inequalities]{Bounding Horizon Area by Angular Momentum, Charge, and Cosmological Constant in 5-Dimensional Minimal Supergravity}

\author{Aghil Alaee}
\address{Department of Mathematics\\
University of Toronto\\
Toronto, ON M5S 2E4, Canada }
\email{a.alaeekhangha@utoronto.ca}

\author{Marcus Khuri}
\address{Department of Mathematics, Stony Brook University, Stony Brook, NY 11794, USA}
\email{khuri@math.sunysb.edu}
	
\author{Hari Kunduri}
\address{Department of Mathematics and Statistics,
McMaster University, Hamilton, ON L8S 4K1,
Canada (on sabbatical leave from Memorial University of Newfoundland)}
\email{hkkunduri@mun.ca}

\thanks{A. Alaee acknowledges the support of a NSERC Postdoctoral Fellowship 502873 and PIMS Postdoctoral Fellowship. M. Khuri acknowledges the support of NSF Grants DMS-1308753 and DMS-1708798. H. Kunduri acknowledges the support of NSERC Grant 418537-2012.}

\begin{abstract}
We establish a class of area-angular momentum-charge inequalities satisfied by stable marginally outer trapped surfaces in 5-dimensional minimal supergravity which admit a $U(1)^2$ symmetry.  A novel feature is the fact that such surfaces can have the nontrivial topologies $S^1 \times S^2$ and $L(p,q)$.  In addition to two angular momenta, they may be characterized by `dipole charge' as well as electric charge. We show that the unique geometries which saturate the inequalities are the horizon geometries corresponding to extreme black hole solutions. Analogous inequalities which also include contributions from a positive cosmological constant are also presented.
\end{abstract}

\maketitle

\section{Introduction}\label{Sec0}

There has been significant progress in establishing sharp geometric inequalities, motivated in part by black hole thermodynamics, which relate the area $A$, angular momenta $\mathcal{J}$, and charge $Q$ of axisymmetric stable marginally outer trapped surfaces (MOTS) in spacetimes satisfying an appropriate energy condition \cite{dain2012geometric,DainClement}. In spacetime dimension $D=4$ a typical example of such an inequality \cite{clement2012proof} is given by
\begin{equation}\label{4dineq}
A \geq 4\pi \sqrt{4\mathcal{J}^2 + Q^4},
\end{equation}
where equality is achieved if and only if the induced geometry of the MOTS arises from a spatial cross section of the event horizon of the extreme Kerr-Newman black hole.
This class of results has been extended to include contributions from a positive cosmological constant \cite{Bryden:2016frb,clement2015area}, and studied in the setting of Einstein-Maxwell-axion-dilaton gravity \cite{FajmanSimon,Rogatko,Yazadjiev_area}. They have also been used to find lower bounds for horizon area in terms of (ADM) mass, angular momentum, and charge \cite{dain2013lower}.

Given the significant interest in black hole solutions in spacetime dimension $D>4$ \cite{emparan2008black}, chiefly motivated from the physical point of view by string theory,  a natural problem is to generalize such inequalities to this setting.  As is well known, $D=5$ asymptotically flat black holes arise as certain intersecting configurations of D-branes, which are dynamical extended objects in string theory.  A classic achievement of string theory is the calculation of the Bekenstein-Hawking entropy $S = A/4$ for a large class of extremal 5-dimensional black holes from a quantum statistical counting of such configurations \cite{Strominger1996}.
Inequalities relating the area, angular momenta, and charge of dynamical black holes can be translated into corresponding relations on the quantum numbers that characterize the string states.

This program has been initiated in the work of Hollands \cite{Hollands2012} (see also \cite{Yazadjiev_area1}), who proved an extension of \eqref{4dineq} to $D > 4$ for vacuum spacetimes, possibly with a positive cosmological constant.  The $(D-2)$-dimensional MOTS $\mathcal{B}$ was assumed to admit a $U(1)^{D-3}$ isometry group.  This requirement implies that $\mathcal{B}$ must be diffeomorphic to $S^3 \times T^{D-5}$, $S^1 \times S^2 \times T^{D-5}$, or $L(p,q) \times T^{D-5}$ where $T^n$ is the $n$-dimensional torus \cite{hollands2012black}. In particular the elegant inequality
\begin{equation}\label{5dvacineq}
A \geq 8\pi \sqrt{|\mathcal{J}_+ \mathcal{J}_-|}
\end{equation}
is shown to hold for all stable $\mathcal{B}$, where $\mathcal{J}_\pm = \mathcal{J}_i v^i_{\pm}$ are certain linear combinations of angular momenta $\mathcal{J}_i$, associated to each $U(1)$ generator, and $v^i_\pm$ are a set of integers which determine the topology of $\mathcal{B}$. As before, the unique geometries that saturate the inequality are the extreme horizon geometries corresponding to each of the allowed topologies. The possible vacuum horizon geometries are completely classified and are in fact known explicitly in closed form \cite{hollands2010all,Kunduri2009}.

It is worthwhile to elaborate on this point. The term near-horizon geometry refers to the precise notion of the spacetime geometry in a neighborhood of a degenerate Killing horizon (for a comprehensive review, see \cite{Kunduri:2013ana}).  For example, the near-horizon geometry associated to the extreme Reissner-Nordstr\"{o}m horizon is a product metric on AdS$_2 \times S^2$, while the near-horizon geometry associated to the extreme Kerr horizon is a twisted $S^2$ bundle over AdS$_2$.  It is important to note that different extreme black holes can have the same associated near-horizon geometry (see \cite{Kunduri:2007vf} for an explicit example in $D=5$).  A spatial cross section of the event horizon (which is a MOTS) is a $(D-2)$-Riemannian manifold embedded in the $D$-dimensional Lorentzian near-horizon spacetime. Therefore, when stating the rigidity results for the area inequalities satisfied by MOTS, we must state that those saturating the inequality are the extreme horizon geometries induced from a near-horizon geometry, rather than a particular extreme black hole solution; indeed, there could be more than one extreme black hole that gives rise to the same induced geometry on $\mathcal{B}$.  Of course, in $D=4$, an axisymmetric extreme electrovacuum black hole must be an extreme member of the Kerr-Newman family \cite{Chrusciel2010}, and so we can state the rigidity result simply in terms of the induced metric on the horizon of an extreme black hole solution without reference to its near-horizon geometry.

The purpose of the present work is to establish an extension of \eqref{5dvacineq} valid for 5-dimensional black holes which carry charges sourced by a Maxwell field $F$. The simplest relevant theory for this purpose is minimal $D=5$ supergravity. As explained in \cite{Alaee:2016jlp}, this theory admits a harmonic map formulation for stationary $U(1)^{2}$-invariant solutions $(\mathbf{g},F)$ which plays a key role in establishing the relevant geometric inequalities. Moreover, all explicitly known charged 5-dimensional black holes (e.g. the charged Myers-Perry solution, the natural generalization of Kerr-Newman) are solutions of supergravity; these solutions will serve as model maps in the construction of the proof.  An added motivation is that it is this theory, and not standard Einstein-Maxwell theory, that arises as a consistent reduction of the ten or eleven-dimensional supergravity theories that govern the low-energy dynamics of string theory.  Therefore minimal supergravity is natural to consider for a number of reasons.

The key difference between the supergravity setting and the pure vacuum case analyzed in  \cite{Hollands2012} is that the space of extreme black hole horizons is significantly larger.  As we will show, one can produce a lower bound for the area of admissible MOTS with fixed angular momenta $\mathcal{J}_i$ and charge $Q$ in terms of an `area functional', which is in turn a certain renormalized Dirichlet energy for singular maps taking $[-1,1] \to G_{2(2)} / SO(4)$; here $G_{2(2)}$ refers to the noncompact real Lie group whose complexification is $G_2$, and the notation $2(2)$ refers respectively to the rank and character of the group.  The critical points of this functional are simply the harmonic maps corresponding to horizon geometries of $U(1)^2$-invariant extreme black holes with the same $\mathcal{J}_i$ and $Q$.  In contrast to the $D=5$ vacuum case, however, a complete classification of all allowed extreme horizon geometries is an open problem (see \cite{kunduri2011constructing} for a partial classification). Indeed, for fixed horizon topology $\mathcal{B}$, one can have distinct families of extreme horizon geometries.  For example, there are non-isometric families of extreme black ring horizon geometries ($\mathcal{B} = S^1 \times S^2$).  In particular, these distinct families have different expressions for the area in terms of conserved charges.  Nevertheless, for a given topology
it is possible to identify a unique extreme horizon geometry by specifying the angular momenta, electric charge, and so-called `dipole charge'.

It is worth emphasizing how this is qualitatively different from the $D=4$ Einstein-Maxwell case.  For fixed angular momentum and electric charge, the unique axisymmetric extreme horizon geometry is that of the extreme Kerr-Newman black hole \cite{chrusciel2012stationary}.  This fact underlies the single inequality \eqref{4dineq}.  In our case, rather than establishing a single unified inequality \eqref{5dvacineq} valid for all $\mathcal{B}$, we will have different inequalities which depend on both the topology of $\mathcal{B}$ and the range of parameters associated with conserved charges.

\section{Description of Main Results}\label{Sec1}

Consider a five dimensional spacetime $(M,\mathbf{g},F)$ where $M$ is a smooth oriented manifold, $\mathbf{g}$ is a Lorentzian metric with signature $(-,+,+,+,+)$, and $F$ is a closed 2-form representing a Maxwell field. Assuming that $F=d\mathcal{A}$, the action for $D=5$ minimal supergravity is given by
\begin{equation}\label{sugraaction}
\mathcal{S} = \int_M (R-12{\Lambda}) \star 1 - \frac{1}{2} F \wedge \star F - \frac{1}{3\sqrt{3}} F \wedge F \wedge \mathcal{A},
\end{equation}
where $\star$ is the Hodge dual operator associated to $\mathbf{g}$ and $\Lambda \geq 0$ is the cosmological constant. The field equations are then expressed as
\begin{equation}\label{field}
\begin{split}
&R_{ab} = \frac{1}{2} F_{ac} F_{b}^{~ c} - \tfrac{1}{12} |F|^2 \mathbf{g}_{ab}+4{\Lambda} \mathbf{g}_{ab},  \\
&d \star F + \frac{1}{\sqrt{3}} F \wedge F=0.
\end{split}
\end{equation}
Note that in contrast to pure Einstein-Maxwell theory, $d \star F \neq 0$. If $H_2(M)\neq 0$ then $\mathcal{A}$ appearing in the action is not globally defined and must be constructed from local potentials.
	
Recall that a marginally outer trapped surfaces (MOTS) is a $3$-dimensional spacelike submanifold $\mathcal{B}$ embedded in the spacetime $(M,g,F)$ with $\theta_{\mathbf{n}}=0$. Here $\theta_{\mathbf{n}}$ is the expansion with respect to the future pointing outward null normal $\mathbf{n}$ and is defined by
\begin{equation}
\theta_{\mathbf{n}} \epsilon_{\gamma}=\epsilon_{\gamma}\text{div}_{\gamma}\mathbf{n}
=\mathfrak{L}_{\mathbf{n}}\epsilon_{\gamma} =\left( \frac{1}{2} \gamma^{ab} \mathfrak{L}_{\mathbf{n}} \gamma_{ab} \right)\epsilon_\gamma,
\end{equation}
where $\mathfrak{L}$ denotes Lie differentiation,
\begin{equation}\label{eq3.2}
\gamma_{ab}=2\mathbf{l}_{(a}\mathbf{n}_{b)}+\mathbf{g}_{ab}
\end{equation}
is the induced metric on $\mathcal{B}$ with volume form $\epsilon_{\gamma}$, and  $\mathbf{l}$ is the future pointing inward null normal such that $\mathbf{g}(\mathbf{n},\mathbf{l})=-1$. The MOTS will be referred to as stable if $\mathfrak{L}_{\mathbf{l}}\theta_{\mathbf{n}}\leq 0$.

The total electric charge contained within $\mathcal{B}$ is given by
\begin{equation}\label{defcharge}
Q =  \frac{1}{16\pi} \int_{\mathcal{B}}\left(\star F + \frac{1}{\sqrt{3}} \mathcal{A} \wedge F\right).
\end{equation}
Inclusion of the second term in the integrand is motivated by the fact that, as a consequence of the Maxwell equation in \eqref{field}, the full integrand is a closed 3-form.  If, in addition $H_2 (\mathcal{B})$ is non-trivial (e.g. $\mathcal{B}= S^1 \times S^2$), a `dipole charge' may be defined by
 \begin{equation}\label{defdipole}
\mathcal{D}[\mathcal{C}] = \frac{1}{2\pi} \int_{\mathcal{C}} F
\end{equation}
for each homology class $[\mathcal{C}]\in H_2(\mathcal{B})$. From the $D=4$ setting, this may seem reminiscent of the magnetic charge, however in $D=5$ it turns out that there is no natural notion of a conserved magnetic charge \cite{Alaee:2016jlp}. Note that if $\mathcal{B} = S^3$ (or indeed any lens space), $H_2(\mathcal{B})$ is trivial.

In order to define a suitable notion of angular momenta, let $\eta_{(i)}$, $i=1,2$ denote the Killing fields with orbits of period $2\pi$ that generate the $U(1)^2$ isometry, so that
\begin{equation}
\mathfrak{L}_{\eta_{(i)}} \mathbf{g} = 0\;,\qquad \mathfrak{L}_{\eta_{(i)}} F = 0.
\end{equation}  	
The angular momentum associated with the generator $\eta_{(i)}$ is then defined by
\begin{equation}\label{defJ}
\mathcal{J}_{i} = \frac{1}{16\pi} \int_\mathcal{B} \star d[\mathbf{g}(\eta_{(i)},\cdot)] + \mathcal{A}(\eta_{(i)}) \left(\star F + \frac{2}{3\sqrt{3}} \mathcal{A} \wedge F \right).
\end{equation}
The first term of the integrand comes from the standard Komar integral, and the remaining
terms are then appended in order to obtain a closed 3-form yielding a conserved quantity.
This is the spacetime version of the definition in terms of initial data used for the proof of the mass-angular momentum-charge inequality \cite{Alaee:2016jlp}.  Moreover, when $F=0$ this reduces to the definition of angular momenta used in the proof of the vacuum inequality \eqref{5dvacineq}. Note that there is an $SL(2,\mathbb{Z})$ freedom in choosing a basis for the $U(1)^2$ generators $\eta_{(i)}$, and hence to define the two angular momenta.

\begin{theorem}\label{Thm1}
Let $(M,\mathbf{g},F)$ be a bi-axisymmetric solution of 5-dimensional minimal supergravity with $\Lambda=0$. If $\mathcal{B}$ is a bi-axisymmetric stable MOTS diffeomorphic to $S^3$ then
\begin{equation}
A\geq 8\pi\sqrt{\left\vert\mathcal{J}_1\mathcal{J}_2+\frac{4Q^3}{3\pi\sqrt{3}}\right\vert},
\end{equation}
and equality holds if and only if $(\mathcal{B},\gamma,F)$ arises from the near-horizon geometry of an extreme charged Myers-Perry black hole.
\end{theorem}

It is important to note that the rigidity statement does not imply that the harmonic map data arising from $(\mathcal{B},\gamma,F)$ agree with that of the specified near-horizon geometry. Rather, by stating that the given data `arise' from the near-horizon geometry of an extreme charged Myers-Perry black hole, we mean that the given data are related to this near-horizon geometry through an isometry in the target symmetric space $G_{2(2)}/SO(4)$.
The same interpretation applies to the remaining theorems of this section.

We remark that in the pure vacuum case $Q=0$ and the above inequality reduces to \eqref{5dvacineq}, whereas if either of the independent angular momenta $\mathcal{J}_i$ vanish then
\begin{equation}
A \geq 16 \sqrt{\frac{\pi |Q|^3}{3\sqrt{3}}},
\end{equation}
which is saturated if and only if the MOTS arises from the near-horizon geometry of the extreme Reissner-Nordstr\"{o}m black hole.

\begin{theorem}\label{Thm2}		
Let $(M,\mathbf{g},F)$ be a bi-axisymmetric solution of 5-dimensional minimal supergravity with $\Lambda=0$. Let $\mathcal{B}$ be a bi-axisymmetric stable MOTS diffeomorphic to $S^1 \times S^2$ with $\mathcal{J}_1$, $\mathcal{J}_2$ representing the angular momentum associated with $S^1$, $S^2$ respectively.
\begin{enumerate}[(a)]
\item If $Q=0$ and $\mathcal{J}_2=0$ then
\begin{equation}
A\geq 4\pi \sqrt{ \frac{\pi |\mathcal{J}_1 \mathcal{D}^3|}{3\sqrt{3}}},
\end{equation}
and equality holds if and only if $(\mathcal{B},\gamma,F)$ arises from the near-horizon geometry of an extreme singly-spinning dipole black ring.

\item If $Q=0$ and $\mathcal{J}_2^2\geq \frac{\pi}{12\sqrt{3}}\mathcal{J}_1\mathcal{D}^3$ then
\begin{equation}
A\geq 8\pi\sqrt{\mathcal{J}^2_{2}-\frac{\pi}{12\sqrt{3}}\mathcal{J}_1\mathcal{D}^3},
\end{equation}
and equality holds if and only if $(\mathcal{B},\gamma,F)$ arises from the near-horizon geometry of an extreme magnetic boosted Kerr string.
\end{enumerate}
\end{theorem}

\begin{theorem}\label{Thm3}
Let $(M,\mathbf{g},F)$ be a bi-axisymmetric solution of 5-dimensional minimal supergravity with $\Lambda=0$. If $\mathcal{B}$ is a bi-axisymmetric stable MOTS diffeomorphic to the lens space $L(p,1)$, with $\mathcal{J}_1 = -\mathcal{J}_2=\mathcal{J}$ and $4Q^3\geq 3p\sqrt{3} \pi\mathcal{J}^2$ then
\begin{equation}
A\geq 8\pi \sqrt{\frac{4pQ^3}{3\sqrt{3} \pi} - p^2\mathcal{J}^2},
\end{equation}
and equality holds if and only if $(\mathcal{B},\gamma,F)$ arises from the near-horizon geometry of an extreme supersymmetric black lens solution \cite{kunduri2014supersymmetric,Tomizawa:2016kjh}.
\end{theorem}

These three theorems yield area-angular momentum-charge inequalities for each of the
possible topologies associated with bi-axisymmetric stable MOTS. It should be noted that several of the results require certain restrictions on the parameters found within the inequalities. These
restrictions arise from the particular nature of the known near-horizon geometries on which the inequalities are modeled. Our method of proof is sufficiently robust that should new near-horizon geometries be found, an immediate consequence would be new area-angular momentum-charge inequalities for stable MOTS with the same topology. Thus modulo the
classification of near-horizon geometries for $D=5$ minimal supergravity, the techniques of this paper are able to produce all possible inequalities of this type.

The following result differs from those above in that it includes contributions from
a cosmological constant $\Lambda\geq 0$ within the inequality. We are only able to treat
the case of spherical topology in this context due to the lack of known explicit solutions with other topologies; in fact it has been shown that ring type de Sitter near-horizon geometries do not exist in vacuum \cite{KhuriWoolgar}. Restrictions on the parameters are needed here only to simplify the expression of the inequality. Indeed, our proof is valid for the full range of parameters, however a precise statement of the inequality in this generality is too unwieldy.

\begin{theorem}\label{thm5}
Let $(M,\mathbf{g},F)$ be a bi-axisymmetric solution of 5-dimensional minimal supergravity with $\Lambda>0$, and let $\mathcal{B}$ be a bi-axisymmetric stable MOTS diffeomorphic to $S^3$.
\begin{enumerate}[(a)]
\item If $\mathcal{J}:=\pm\mathcal{J}_i$ for $i=1,2$ and $Q=0$ then
\begin{equation}\label{mainJin}
 \frac{\Lambda^3A^6}{2^{10}\pi^6\mathcal{J}^2}\leq
 \frac{\left(A\sqrt{A^2+512\pi^2 \mathcal{J}^2}-A^2-128\pi^2\mathcal{J}^2\right)^3}
 {\left(A-\sqrt{A^2+512\pi^2 \mathcal{J}^2}\right)^4},
\end{equation}
and equality holds if and only if $(\mathcal{B},\gamma,F)$ arises from the near-horizon geometry of an extreme Chong-Cvetic-Lu-Pope (CCLP) black hole with positive cosmological constant.

\item  If $\mathcal{J}_1=\mathcal{J}_2=0$ then
\begin{equation}
64\pi^2Q^2\leq	12\left(\frac{A\pi}{2}\right)^{4/3}-6\Lambda A^2,
\end{equation}
and equality holds if and only if $(\mathcal{B},\gamma,F)$ arises from the near-horizon geometry of an extreme Reissner-Nordstr\"{o}m-de Sitter black hole.
\end{enumerate}
\end{theorem}

\section{Construction of Potentials and Relation to Conserved Charges}\label{Sec2}

In this section we construct scalar potentials, and demonstrate how these potentials encode the charges and angular momenta defined above. The procedure follows that given in \cite{kunduri2011constructing} specialized to the case of the theory \eqref{sugraaction}.

Observe that since $dF =0$, Cartan's formula $\mathfrak{L}_X = \iota_X d + d \iota_X$ may be used to show that the following 1-forms are closed, yielding the existence of scalar potentials satisfying
\begin{equation}\label{defmagnetic}
d\psi^i=\iota_{\eta_{(i)}} F,
\end{equation}
where $\iota$ denotes the operation of interior product.
These may be interpreted as magnetic potentials and are globally defined in a tubular neighborhood $\tilde{M}$ of $\mathcal{B}$, since the orbit space $\tilde{M}/U(1)^2$ is simply connected \cite{Hollands2008} and the potentials
are functions on the orbit space. To see this last point, note that the quantities
\begin{equation}
\mathfrak{L}_{\eta_{(i)}}\psi^j = \iota_{\eta_{(i)}} \iota_{\eta_{(j)}} F
\end{equation}
are constants by standard arguments, and since the $\eta_{(i)}$ vanish at the
rotation axes these constants are zero.

Now define the 1-form
\begin{equation}\label{eq2.3}
\Upsilon = -\iota_{\eta_{(1)}} \iota_{\eta_{(2)}} \star F
\end{equation}
and observe that
\begin{equation}\label{eq2.4}
d\Upsilon = \frac{1}{\sqrt{3}}\iota_{\eta_{(1)}} \iota_{\eta_{(2)}}d\left(\mathcal{A} \wedge F\right)=\frac{1}{\sqrt{3}} d\left(\psi^1 d\psi^2 - \psi^2 d\psi^1 \right).
\end{equation}
This implies the existence of an electric potential satisfying
\begin{equation}\label{defchi}
d\chi = \Upsilon - \frac{1}{\sqrt{3}}\left(\psi^1 d\psi^2 - \psi^2 d\psi^1 \right).
\end{equation}
With the same reasoning as above, it may be shown that this potential is also globally
defined.

In order to construct charged twist potentials for the angular momentum consider the 1-forms
\begin{equation}\label{eq2.8}
\Theta^i = \star(\eta_{(1)} \wedge \eta_{(2)} \wedge d\eta_{(i)}),
\end{equation}
which satisfy
\begin{equation}\label{eq2.9}
d\Theta^i  = 2\star(\eta_{(1)} \wedge \eta_{(2)} \wedge \text{Ric}(\eta_{(i)}))
\end{equation}
where $\text{Ric}$ denotes the Ricci tensor of the spacetime metric $\mathbf{g}$. With the help of the Einstein equations \eqref{field}, an involved calculation \cite{kunduri2011constructing} shows that
\begin{equation}\label{eq2.10}
d\Theta^i =-\Upsilon\wedge \iota_{\eta_{(i)}} F
= d\left[\psi_i \left(d\chi + \frac{1}{3\sqrt{3}}(\psi^1 d\psi^2 - \psi^2 d\psi^1)\right)\right].
\end{equation}
It follows that there exist globally defined twist potentials such that
\begin{equation}\label{eq2.11}
d\zeta^i = \Theta^i  - \psi^i \left[d\chi + \frac{1}{3\sqrt{3}}(\psi^1 d\psi^2 - \psi^2 d\psi^1)\right].
\end{equation}

It will now be shown how these potentials are related to the various charges associated
with the MOTS $\mathcal{B}$. Since $\mathcal{B}$ is bi-axisymmetric
the isometry generators $\eta_{(i)}$ are tangent to $\mathcal{B}$. We may then introduce $2\pi$-periodic angular coordinates $\phi^i$ on $\mathcal{B}$ adapted to the symmetries, so that $\eta_{(i)} = \partial / \partial \phi^i$.  A third coordinate function $x$ arises from the volume form by
\begin{equation}
dx = \mathcal{C} \text{Vol}_{\gamma} (\eta_{(1)}, \eta_{(2)}, \cdot),
\end{equation}
where $\mathcal{C}$ is a constant.  According to \cite{Hollands2008} the 1-dimensional orbit space $\mathcal{B} / U(1)^2$ is diffeomorphic to a closed interval, and the constant $\mathcal{C}$ is chosen so that the orbit space is parameterized by $x \in [-1,1]$. In order to compute
the electric charge in terms of the potential $\chi$, observe that if
$\omega$ is a 3-form on $\mathcal{B}$ then
\begin{equation}
\int_\mathcal{B} \omega = 4\pi^2 \int_{-1}^{1} \iota_{\eta_{(2)}}\iota_{\eta_{(1)}} \omega.
\end{equation}
Then using the definition \eqref{defcharge}, \eqref{eq2.3}, \eqref{eq2.4}, and \eqref{defchi} that
\begin{equation}
Q =\frac{\pi}{4} \int_{-1}^{1} d\chi= \frac{\pi}{4}\left(\chi(1) - \chi(-1)\right).
\end{equation}

Next suppose that $\mathcal{B}= S^1 \times S^2$ and $\int_{S^2} F \neq 0$, so that the vector potential $\mathcal{A}$ is not globally defined.   From \eqref{defmagnetic} it follows that $F = d\phi^i \wedge d\psi^i$, and hence
\begin{equation}
\mathcal{D} = \frac{1}{2\pi} \int_{S^2} F = v^i(\psi^i(-1) - \psi^i(1)),
\end{equation}
where $v^i \eta_{(i)}$ ($v^i \in \mathbb{Z}$) is the Killing field that vanishes at the poles of the $S^2$.
			
Finally we turn to the angular momenta \eqref{defJ}. First note that
\begin{equation}
\iota_{\eta_{(2)}}\iota_{\eta_{(1)}}\left[\mathcal{A}(\eta_{(i)}) \left(\star F + \frac{2}{3\sqrt{3}} \mathcal{A} \wedge F \right)\right] = -\psi^i \left[d\chi + \frac{1}{3\sqrt{3}}\left(\psi^1 d\psi^2 - \psi^2 d\psi^1 \right)\right].
\end{equation}
Moreover, the formula
\begin{equation}
\iota_X \star \varsigma = (-1)^{4-k} \star(X \wedge \varsigma)
\end{equation}
is valid for $k$-forms $\varsigma$ on 5-dimensional Lorentzian manifolds and
reveals that
\begin{equation}
\iota_{\eta_{(2)}}\iota_{\eta_{(1)}}\star d [\mathbf{g}(\eta_{(i)}, \cdot)] = \Theta^i.
\end{equation}
Altogether this yields
\begin{equation}
\mathcal{J}_i = \frac{\pi}{4} \int_{-1}^{1} \left(\Theta^i - \psi^i\left[d\chi + \frac{1}{3\sqrt{3}}\left(\psi^1 d\psi^2 - \psi^2 d\psi^1 \right)\right]\right)  = \frac{\pi}{4} \int_{-1}^{1} d\zeta^i = \frac{\pi}{4}\left( \zeta^i (1) - \zeta^i(-1)\right).
\end{equation}

\section{The Area Functional}
\label{Sec3}

We now turn to deriving a lower bound on the area of a bi-axisymmetric stable MOTS $\mathcal{B}$ in terms of a certain area functional.  The critical points of this functional will be shown to correspond to spacetimes that describe, in a precise sense, the geometry in a neighborhood of a (stationary) extreme black hole. These near-horizon geometries are solutions of the spacetime Einstein equations in their own right, and will play the role of minimizers in what follows.

In the previous section coordinates $(x,\phi^1,\phi^2)$ where introduced on $\mathcal{B}$
in which the $\phi^i$ are adapted to the $U(1)^2$ isometry and $x$ parameterizes the orbit
space $\mathcal{B}/U(1)^2\cong[-1,1]$. As in \cite{Hollands2012} the induced metric on $\mathcal{B}$ takes the following form when expressed in these coordinates
\begin{equation}\label{eq2.12}
\gamma_{mn}dy^m dy^n=\frac{ dx^2}{\mathcal{C}^2\det{\lambda}}+{\lambda}_{ij}d\phi^id\phi^j,
\end{equation}
where the constant $\mathcal{C}>0$ has length dimension $-3$ and is related to the area of $\mathcal{B}$ by
\begin{equation}\label{eq2.13}
A=8\pi^2 \mathcal{C}^{-1}.
\end{equation}
The topology of $\mathcal{B}$ is characterized by the integer linear combinations of Killing fields that vanish at the endpoints $x = \pm 1$, which represent the fixed points of the torus action. Suppose that $a_\pm^i \eta_{(i)} \to 0$ as $x \rightarrow \pm 1$, with $a^i_\pm \in \mathbb{Z}$.  The matrix ${\lambda}_{ij}$ is rank 2 for $x \in (-1,1)$ and has a 1-dimensional kernel at $x = \pm 1$ spanned by $a_\pm ^i$, that is ${\lambda}_{ij}a^i_{\pm} \to 0$ at the endpoints.  Without loss of generality it may be assumed that $a_+ = (1,0)$, $a_- = (q,p)$ for some $p,q \in \mathbb{Z}$ with $\mathrm{gcd}(p,q) = 1$.  We have $(q,p) = (0,\pm 1)$ for $S^3$ topology, $(q,p) = (\pm 1, 0)$ for $S^1 \times S^2$ topology, and $L(p,q)$ otherwise \cite{Hollands2012}.   Note that the absence of  conical singularities requires
\begin{equation}\label{eq2.14}
\lim_{x\to \pm 1} \frac{(1-x^2)^2}{\det {\lambda} \cdot a^i_\pm a^j_\pm {\lambda}_{ij}} = \mathcal{C}^2.
\end{equation}

Following \cite{Hollands2012,hollands2010all,racz2008simple}, Gaussian null coordinates $(u,r,y^m)$ may be introduced in a neighborhood of the MOTS $\mathcal{B}$. Here $\mathbf{n}=\partial_u$ and $\mathbf{l}=\partial_r$
are future pointing null vectors which coincide with the normal vectors of the same notation in Section \ref{Sec1} on $\mathcal{B}$, and satisfy $\mathbf{g}(\mathbf{n},
\mathbf{l})=-1$. The coordinates $y^m$ are Lie transported off of $\mathcal{B}$ by $\mathbf{n}$ and $\mathbf{l}$. This process yields a foliation of the neighborhood of $\mathcal{B}$, with parameters $(u,r)$, whose leaves are denoted by $\mathcal{B}(u,r)$ and for which $\mathcal{B}(0,0)=\mathcal{B}$. It can be shown that in these coordinates the spacetime metric takes the Gaussian null form
\begin{equation}\label{eq3.1}
\mathbf{g}=-2du\left(dr-\alpha r^2 du-r\beta_m dy^m\right)+\gamma_{mn}dy^m dy^n,
\end{equation}
where $\alpha$ is a smooth function, $\beta=\beta_m dy^m$ is a 1-form, and $\gamma$ is the induced metric on $\mathcal{B}(u,r)$. Note that this expression may be
simplified with the help of the coframe
\begin{equation}\label{eq2.16}
e^+ = du, \quad e^- = dr - \alpha r^2 du - r(\beta_x dx + \beta_i d\phi^i),\quad
e^x = \frac{dx}{\mathcal{C} \sqrt{\det{\lambda}}}, \quad e^i = d\phi^i,
\end{equation}
so that
\begin{equation}\label{eq2.17}
\mathbf{g} = -2 e^+ e^- + (e^x)^2+ \lambda_{ij} e^i e^j.
\end{equation}

\begin{lemma}\label{lemma3.1}
Let $(M, \mathbf{g}, F)$ be a bi-axisymmetric solution of 5-dimensional minimal supergravity, and let $\mathcal{B}$ be a bi-axisymmetric stable MOTS.
For any bi-axisymmetric $\varphi\in C^{\infty}(\mathcal{B})$ the stability inequality holds
\begin{equation}\label{eq3.5}
\int_{\mathcal{B}}2|\nabla\varphi|^2_{\gamma}+ \left(R_{\gamma}
+\frac{1}{2}\left[ \langle\beta,N\rangle_{\gamma}^2-|\beta|^2_{\gamma}\right]
-2T(\mathbf{n},\mathbf{l})-12\Lambda\right)\varphi^2 \geq 0,
\end{equation}
where $R_{\gamma}$ is the scalar curvature of $\mathcal{B}$, $N=\mathcal{C}\sqrt{\det\lambda}\partial_{x}$ is the unit normal to the Killing directions $\eta_{(i)}$, and $T$ denotes the stress-energy tensor.
\end{lemma}

\begin{proof}
A computation \cite{Hollands2007} shows that
\begin{equation}\label{eq3.1}
R_{\gamma}-\text{div}_{\gamma}\beta-\frac{1}{2}|\beta|^2_{\gamma}-2T(\mathbf{n},
\mathbf{l})
-12\Lambda=-2\theta_{\mathbf{n}}\theta_{\mathbf{l}}-2\mathfrak{L}_{\mathbf{l}}
\theta_{\mathbf{n}}.
\end{equation}
Since $\mathcal{B}$ is a stable MOTS $\theta_{\mathbf{n}}=0$ and $\mathfrak{L}_{\mathbf{l}}\theta_{\mathbf{n}}\leq 0$. It follows that
\begin{equation}\label{eq3.2}
R_{\gamma}-\text{div}_{\gamma}\beta-\frac{1}{2}|\beta|^2_{\gamma}-2T(\mathbf{n},
\mathbf{l})-12\Lambda\geq 0.
\end{equation}
Then multiplying by $\varphi^2$, integrating the divergence term by parts, and applying Young's inequality yields the desired result.
\end{proof}

We now seek to express the integrand of \eqref{eq3.5} in terms of the potentials of Section \ref{Sec2} and the fiber metric $\lambda$. Let $\beta_i=\beta(\eta_{(i)})$ and $\beta^i=\lambda^{ij}\beta_{j}$, then a calculation gives
\begin{align}\label{eq2.18}
\begin{split}
\eta_{(1)} \wedge \eta_{(2)} \wedge d\eta_{(i)} =& \beta_i(-e^+ \wedge e^- - r e^+ \wedge \beta) \wedge ( \det{\lambda} e^1) \wedge e^2\\
&- r {\lambda}_{ij} (\partial_x \beta^j) \det {\lambda} e^+ \wedge (\mathcal{C} \sqrt{\det {\lambda}} e^x) \wedge e^1 \wedge e^2.
\end{split}
\end{align}
Evaluating at $r =0$ yields
\begin{equation}\label{eq2.19}
\beta_i=\mathcal{C}\Theta^{i}_{x},
\end{equation}
where
\begin{equation}\label{Thetax}
\Theta^i_x = \Theta^i(\partial_{x}) = \partial_x\zeta^i+ \psi^i \left(\partial_x\chi + \frac{1}{3\sqrt{3}}(\psi^1 \partial_x\psi^2 - \psi^2 \partial_x\psi^1)\right).
\end{equation}
It follows that
\begin{equation}\label{eq2.20}
|\beta|^2_{\gamma}=\left<\beta,N\right>^2_{\gamma}+{\lambda}^{ij}\beta_i\beta_j
=\left<\beta,N\right>^2_{\gamma}+\mathcal{C}^2{\lambda}^{ij}\Theta^i_x\Theta^j_x.
\end{equation}
Furthermore a computation \cite{Hollands2012} shows that
\begin{equation}\label{eq2.21}
R_{\gamma}=\mathcal{C}^2\det{\lambda}\left[-\frac{\partial^2_x\det{\lambda}}{\det {\lambda}}+\frac{1}{4}\frac{\left(\partial_x \det{\lambda}\right)^2}{(\det {\lambda})^2}-\frac{1}{4}\text{Tr}({\lambda}^{-1}\partial_x{\lambda})^2\right].
\end{equation}
			
We now turn to the Maxwell field in order to compute the relevant portion of the stress-energy tensor. As shown in \cite{kunduri2011constructing}, this field may be expressed as
\begin{equation}\label{eq2.22}
F=\frac{1}{\det{\lambda}}\left[\star\left(\eta_{(2)}\wedge\eta_{(1)}\wedge \Upsilon\right)+(\det{\lambda}) {\lambda}^{ij}\eta_{(i)}\wedge d\psi^j\right].
\end{equation}
Since $\chi$ and $\psi^i$ are functions of $x$ alone, and
\begin{equation}\label{eq3.15}
\Upsilon_x=\Upsilon(\partial_{x})= \partial_x\chi+ \frac{1}{\sqrt{3}}\left(\psi^1 \partial_x\psi^2 - \psi^2 \partial_x\psi^1 \right),
\end{equation}			
it follows that
\begin{equation}\label{eq2.23}
F = -\mathcal{C}\Upsilon_x  du \wedge e^-  - r \mathcal{C} \Upsilon_x \beta_i du \wedge d\phi^i  + r\beta^i \partial_x \psi^i du \wedge dx - \partial_x \psi^i dx \wedge d\phi^i.
\end{equation}
Next note that
\begin{equation}
T_{ab} = \frac{1}{8} (\star F)_{acd}(\star F)_b^{~cd} + \frac{1}{4}F_{ac}F_b^{~c}.
\end{equation}
Since $\mathbf{n} = \partial_u$ and $\mathbf{l} = \partial_r$, a computation shows that (at $r=0$)
\begin{equation}
(\iota_{\mathbf{n}} F)_c (\iota_{\mathbf{l}} F)^c = \mathcal{C}^2 \Upsilon_x^2.
\end{equation}
In order to deal with the term involving $\star F$, observe that
\begin{equation}
\star(e^+ \wedge e^- ) = -\text{Vol}_{\gamma} = -\mathcal{C}^{-1} dx \wedge d\phi^1 \wedge d\phi^2,
\qquad  \star(e^x \wedge e^i)  = \epsilon^{i}_{j} e^+ \wedge e^- \wedge e^j,
\end{equation}
where $\epsilon$ is the volume form associated with $\lambda$. From \eqref{eq2.22} we then have
\begin{equation}
\star F = \mathcal{C}\Upsilon_x  \text{Vol}(\gamma) - \mathcal{C}(\partial_x \psi^i) \sqrt{\det {\lambda}} \epsilon_{ij} e^+ \wedge e^- \wedge e^j,
\end{equation}
which implies
\begin{equation}
(\iota_{\mathbf{n}} \star F)_{cd} (\iota_{\mathbf{l}} \star F)^{cd} = 2 \mathcal{C}^2 (\det {\lambda}) \lambda^{ij} \partial_x \psi^i \partial_x \psi^j.
\end{equation}
Therefore at $r=0$
\begin{equation}\label{G(n,l)}
T(\mathbf{n},\mathbf{l}) = \frac{\mathcal{C}^2}{4}\left(\Upsilon_x^2 + (\det \lambda) {\lambda}^{ij} \partial_x \psi^i \partial_x \psi^j\right).
\end{equation}

It remains to choose $\varphi$ and compute its Dirichlet energy density.
Let $\xi\in C^{\infty}(\mathcal{B})$ be a particular smooth positive function (of $x$) associated to the relevant extreme stationary black hole solution of 5D minimal supergravity, which satisfies $\xi=1$ when $\Lambda=0$. Then set
\begin{equation}\label{eq3.8}
\varphi = \sqrt{\xi\frac{(1-x^2)}{\det {\lambda}}}=\sqrt{\xi\bar{\varphi}}.
\end{equation}
Note that
\begin{equation}\label{eq3.9} \nabla\varphi=\frac{\nabla\xi}{2}\sqrt{\frac{\bar{\varphi}}{\xi}}
+\frac{\nabla\bar{\varphi}}{2}
\sqrt{\frac{\xi}{\bar{\varphi}}},\qquad |\nabla\varphi|^2=\frac{|\nabla\xi|^2}{4}{\frac{\bar{\varphi}}{\xi}}
+\frac{|\nabla\bar{\varphi}|^2}{4}
{\frac{\xi}{\bar{\varphi}}}+\frac{1}{2}\nabla\xi\cdot \nabla\bar{\varphi},
\end{equation}
where
\begin{equation}\label{eq3.10}
\frac{|\nabla\bar{\varphi}|^2}{4\bar{\varphi}} = -\mathcal{C}^2 + \frac{\mathcal{C}^2}{1-x^2} + \mathcal{C}^2 x \frac{\partial_x \det {\lambda}}{\det {\lambda}} + \frac{\mathcal{C}^2(1-x^2)}{4} \frac{(\partial_x \det {\lambda})^2}{(\det {\lambda})^2}.
\end{equation}

By combining these formulae we find that the integrand of \eqref{eq3.5} takes the form
\begin{equation}\label{eq3.11}
\begin{split}
	&-\mathcal{C}^2(1-x^2)\xi\left(\frac{1}{4} \frac{(\partial_x \det {\lambda})^2}{(\det {\lambda})^2} + \frac{\text{Tr} ({\lambda}^{-1} \partial_x {\lambda})^2}{4} + \frac{1}{2\det{\lambda}} {\lambda}^{ij} \Theta_{x}^{i} \Theta_{x}^{j} \right.\\
	&\left.+ \frac{1}{2\det {\lambda}}\left[\Upsilon_{x}^{2} + \det \lambda \lambda^{ij} \partial_x \psi^i \partial_x \psi^j\right]\right) + \frac{2\mathcal{C}^2\xi}{1-x^2}-\frac{12\Lambda(1-x^2)\xi}{\det\lambda}\\
	&-\mathcal{C}^2\partial_x\left(\frac{(1-x^2) \partial_x \det {\lambda}}{\det {\lambda}}\xi  + 2x\xi \right)+\mathcal{C}^2(1-x^2)\frac{\xi'^2}{2\xi}.
\end{split}
\end{equation}
Since
\begin{equation}
\det {\lambda}  = c_\pm (1-x^2) + O(1-x^2)^2 \quad \text{ as } \quad x\to \pm 1
\end{equation}
for some constants $c_{\pm}$, it holds that
\begin{equation}\label{eq3.12}
\left(\frac{(1-x^2) \partial_x \det {\lambda}}{\det {\lambda}}\xi  + 2x \xi\right)\Bigg|_{x=-1}^{x=+1} = 0.
\end{equation}
Therefore in light of Lemma \ref{lemma3.1} the following area functional is nonpositive
\begin{equation}\label{eq3.13}
\mathfrak{I} = \int_{-1}^{+1}\xi\left[(1-x^2)I-\frac{1}{1-x^2} \right]\; d x+\int_{-1}^{1}(1-x^2)\left(\frac{6\Lambda\xi}{\mathcal{C}^2\det{\lambda}}
-\frac{\xi'^2}{4\xi}\right)\, dx\leq 0,
\end{equation}
where
\begin{equation}\label{eq3.14}
I=\frac{1}{8} \frac{(\partial_x\det {\lambda})^2}{(\det {\lambda})^2} + \frac{1}{8}\text{Tr} ({\lambda}^{-1}\partial_x {\lambda})^2 + \frac{1}{4\det{\lambda}} \Theta^T_x\lambda^{-1}\Theta_x + \frac{1}{4 \det {\lambda}}\Upsilon^2_x + \frac{1}{4}\partial_x\psi^T\lambda^{-1}\partial_x\psi
\end{equation}
with $\Theta^T_x=(\Theta^1_x,\Theta^2_x)$ and $\psi^T=(\psi^1,\psi^2)$.

\section{Relation to Near-Horizon Geometries of Extreme Black Holes}
\label{Sec4}

In this section the relationship of the area functional $\mathcal{I}$ to a harmonic energy will be described. The latter arises from the reduction of Einstein's equations on $U(1)^2$-invariant spacetimes. In particular, the critical points of this functional give rise to near-horizon geometries. For simplicity, the discussion here will be restricted to the case $\Lambda=0$ and $\xi=1$.

Observe that the functional may be reorganized as
\begin{equation}
\mathfrak{I} = \int_{-1}^{+1}\left[(1-x^2)G_{AB} \frac{dX^A}{dx} \frac{dX^B}{dx}-\frac{1}{1-x^2} \right]\; d x \leq 0,
\end{equation}
where $I$ has been expressed as the pullback to the orbit space $[-1,1]$ of the nonpositively curved metric on symmetric space $G_{2(2)}/SO(4)$ given by
\begin{equation}\label{targetmetric}
 G_{AB}dX^A dX^B = \frac{(d \det \lambda)^2}{8(\det \lambda)^2} + \frac{\text{Tr} (\lambda^{-1} d \lambda)^2}{8} + \frac{\lambda^{ij}\Theta^i  \Theta^j}{4\det\lambda}  + \frac{\Upsilon^2}{4 \det \lambda}
 + \frac{\lambda^{ij} d\psi^i d\psi^j}{4},
\end{equation}
with target space coordinates $X= (\lambda_{ij}, \zeta^i,\chi, \psi^i)$; note that $\Upsilon$ and $\Theta^i$ are given in terms of these coordinates by \eqref{defchi} and \eqref{eq2.11} respectively. Hence $\mathfrak{I}$ is related in a rather simple way to the Dirichlet energy of maps $[-1,1] \to G_{2(2)}/SO(4)$.  Furthermore it turns out that the target metric \eqref{targetmetric} may be given conveniently by
\begin{equation}
G_{AB} dX^A dX^B = \frac{1}{16}\text{Tr}(\mathcal{M}^{-1} d\mathcal{M}
  \mathcal{M}^{-1} d\mathcal{M}),
\end{equation}
where $\mathcal{M}$ is a positive definite, unimodular coset representative of $G_{2(2)}/SO(4)$ constructed from the scalars $X^A$ whose specific form will not be required here (see, e.g. \cite{kunduri2011constructing}).

In what follows it will be shown that $\mathfrak{I}$ vanishes on harmonic maps, and that these harmonic maps arise from near-horizon geometries.
To begin, consider a 5-dimensional spacetime $(M,\mathbf{g},F)$ which admits a $U(1)^2$ isometry subgroup. The metric may be expressed in the general form
\begin{equation}
\mathbf{g} = \frac{h_{\mu \nu} dx^\mu dx^\nu}{\text{det}\tilde{\lambda}}+
\tilde{\lambda}_{ij}(d\tilde{\phi}^i + \omega^i)(d\tilde{\phi}^j + \omega^j)
\end{equation}
where as before $\partial_{\tilde{\phi}^i}$, $i=1,2$ generate the isometry group and $x^\mu$ represent coordinates on a 3-dimensional `base space' $M_3$ with Lorentzian metric $h$.  The $\omega^i = \omega^i_\mu dx^\mu$ are 1-forms on $M_3$ which measure the obstruction of the Killing fields to being hypersurface orthogonal, and $\tilde{\lambda}_{ij}$ are functions on $M_3$.  Thus the spacetime can be viewed as a $T^2$ fibration over $M_3$. In addition, the decomposition of the Maxwell field into scalar potentials has been discussed in Section \ref{Sec3}.

Now suppose that $(M,\mathbf{g},F)$ is a solution of the field equations of minimal supergravity \eqref{field}. Upon reduction it can be shown that the resulting equations describe the critical points of a 3-dimensional theory of gravity coupled to a wave map (nonlinear sigma model) with action
\begin{equation}
\mathcal{S}[h,X] = \int_{M_3}\left(R_h - 2 h^{\mu \nu}G_{AB}\partial_\mu X^A \partial_\nu X^B \right)\text{Vol}_h,
\end{equation}
where $R_h$ is the scalar curvature of $h$.  The reduced field equations are then given by
\begin{align}\label{fieldeqsigma}
\begin{split}
\text{Ric}(h)_{\mu \nu} &= \frac{1}{8} \text{Tr}(\mathcal{M}^{-1}
\partial_\mu \mathcal{M} \mathcal{M}^{-1} \partial_\nu \mathcal{M}),
\\
\nabla^{\mu}(\mathcal{M}^{-1}\partial_\mu \mathcal{M}) &=0.
\end{split}
\end{align}
As is well known, a similar reduction occurs for other gravity models reduced on tori, most notably pure vacuum gravity and $D=4$ Einstein-Maxwell theory.

Let us further assume that the spacetime contains a degenerate Killing horizon. This means that there is an embedded null hypersurface $\mathcal{N}$ on which $|V|=0$ and $\nabla_V V =0$, for some Killing field $V$.
A cross-section of $\mathcal{N}$ is a spatial $3$-dimensional manifold $H$, which  will be taken to be closed.  The most important examples of such spacetimes are extreme stationary black holes with horizon cross-sections $H$. In a neighborhood of $\mathcal{N}$ one may introduce Gaussian null coordinates, and take the near-horizon limit \cite{Kunduri:2013ana} to find the near-horizon metric
\begin{equation}
\mathbf{g}_{NH} = -2 du (dr -r^2 \tilde{\alpha}(y) du -r\tilde{\beta}_m(y) dy^m)
 + \tilde{\gamma}_{mn}(y) dy^m dy^n
\end{equation}
where $\tilde{\alpha}$ and $\tilde{\beta}$ are a smooth function and a 1-form on the 3-dimensional closed manifold $(H, \tilde{\gamma})$.  Note that $V = \partial_u$ and that $\mathcal{N}$ is defined by $r =0$.   Thus the near-horizon geometry is characterized completely by the triple $(\tilde{\alpha}, \tilde{\beta}_m, \tilde{\gamma}_{mn})$, which are collectively referred to as the near-horizon data.

This near-horizon geometry inherits the $U(1)^2$ isometries from its `parent' spacetime. In fact it is shown in \cite{Kunduri:2007vf} that there is an `enhancement of symmetry' from $\mathbb{R} \times U(1)^{2}$ to $SO(2,1) \times U(1)^{2}$.  It follows that the near-horizon metric and Maxwell field take the form
\cite{kunduri2011constructing, Kunduri:2007vf}
\begin{align}
\begin{split}
\mathbf{g}_{NH} =& \Xi(x)\left[-\frac{r^2 du^2}{\ell^2} - 2du dr\right] + L^2\left[\frac{dx^2}{\det \tilde{\lambda}(x)} + \tilde{\lambda}_{ij}(x)\left(d\phi^i + \frac{\mathbf{b}^i rdu}{L^2}\right)\left(d\phi^j + \frac{\mathbf{b}^j rdu}{L^2}\right)\right],   \\
F_{NH} =& d \left[\frac{\mathbf{a} r du}{L} + L \tilde{\psi}_i(x)\left(d\phi^i + \frac{\mathbf{b}^i r du}{L^2}\right) \right].
\end{split}
\end{align}
The constants $\ell$ and $L$ are length scales introduced so that certain coordinates are dimensionless,  $\Xi(x)>0$ and $\tilde{\psi}_i(x)$ are smooth functions on $H$, and $\mathbf{a}$, $\mathbf{b}^i$ are constants.  Observe that the 2-dimensional metric in the first square bracket is that of AdS$_2$.
Hence a near-horizon geometry can be though of as (in general a twisted) $H$ bundle over AdS$_2$. Note that when $\mathbf{b}^i \neq 0$, the action of $SO(2,1)$ will transform $r du$ by an exact function, which can be compensated by a corresponding $U(1)$ shift in the appropriate angular coordinate. It is easily seen that the near-horizon geometries of  extreme Reissner-Nordstr\"{o}m  (AdS$_2 \times S^2$) and the extreme Kerr (a twisted fibration of $S^2$ over AdS$_2$) both fall into the above general class.

The near-horizon data may be identified with harmonic map coordinates in the following way; explicit details are given in \cite{kunduri2011constructing}. Set $\bar{\phi}^i = L \phi^i$ and $(x^1,x^2,x^3) = (v,r,x)$ then
\begin{equation}
h_{\mu \nu} dx^\mu dx^\nu = L^2 dx^2 + \Xi \det \tilde{\lambda}(x)\left[-\frac{r^2 du^2}{\ell^2} - 2du dr\right], \qquad
\omega^i = \frac{\mathbf{b}^i r du}{L}, \qquad \tilde{\lambda}_{ij} = L^{-2} \lambda_{ij}.
\end{equation}
As for the potentials
\begin{equation}
\psi^i = -L \tilde{\psi}_i,
\end{equation}
and
\begin{equation}
\partial_{x}\chi = \frac{L ^2(\mathbf{a} + \mathbf{b}^i \tilde{\psi}_i )}{\Xi} - \frac{L^2}{\sqrt{3}}\left(\tilde{\psi}_1 \partial_x \tilde{\psi}_2
- \tilde{\psi}_2 \partial_x \tilde{\psi}_1 \right).
\end{equation}
Furthermore a calculation shows that
\begin{equation}
\Theta^i = \frac{L^3 \mathbf{b}^j\tilde{\lambda}_{ij} }{\Xi} dx,
\end{equation}
which implies that the charged twist potentials are given by
\begin{equation}
\partial_{x}\zeta^i = \frac{L^3 \mathbf{b}^j\tilde{\lambda}_{ij} }{\Xi} + L \tilde{\psi}_i\left[\partial_{x}\chi  + \frac{L^2}{3\sqrt{3}}\left(\tilde{\psi}_1
\partial_{x}\tilde{\psi}_2 - \tilde{\psi}_2 \partial_{x}\tilde{\psi}_1\right)\right].
\end{equation}
Note also that $\mathcal{C} = L^{-3}$.

In summary, given a near-horizon geometry we can read off the corresponding harmonic map data $(\lambda_{ij},\zeta^i, \chi, \psi^i)$, and the process can clearly be reversed to solve for $(\tilde{\lambda}_{ij},\tilde{\psi}_i, \mathbf{b}^i,  \mathbf{a})$.  It is also evident that the matrix $\mathcal{M}$ defined above is a function of $x$ alone.  Using this, the coupled 3D gravity-harmonic map equations \eqref{fieldeqsigma} may be simplified. The $(uu)$ and $(ur)$ components of the Einstein equations yield
\begin{equation}
\partial_{x}^2(\Xi \det \tilde{\lambda}) + 2 \frac{L^2}{\ell^2} = 0
\end{equation}
so that
\begin{equation}
\Xi \det \tilde{\lambda}(x) = \frac{L^2}{\ell^2} (1-x^2),
\end{equation}
where we have used the fact that $\Xi \det \tilde{\lambda}$ vanishes at $x=\pm 1$ (where $\tilde{\lambda}_{ij}$ has rank 1).  Note that the induced metric on a horizon-cross section $H$ is then
\begin{equation}
\frac{\ell^2 \Xi(x) dx^2}{(1-x^2)} + \lambda_{ij} d\phi^i d\phi^j.
\end{equation}
The harmonic map equations reduce to
\begin{equation}\label{matrixeq}
\partial_{x} \left[(1-x^2) \mathcal{M}^{-1} \partial_{x} \mathcal{M}\right] =0,
\end{equation}
and coincide with the Euler-Lagrange equations for the functional $\mathfrak{I}$. Thus, the near-horizon geometries are critical points of $\mathfrak{I}$. Furthermore, the $(xx)$ component of the 3D Einstein equations \eqref{fieldeqsigma} place an algebraic constraint on $\mathcal{M}$. Namely, direct integration produces
\begin{equation}
(1-x^2)\mathcal{M}^{-1}\partial_{x}\mathcal{M}=\mathcal{M}_0
\end{equation}
for some constant matrix $\mathcal{M}_0$.  Since
\begin{equation}
\text{Ric}(h)_{xx} = \frac{2}{(1-x^2)^2},
\end{equation}
it follows that $\text{Tr}(\mathcal{M}_{0}^2) = 16$.  The remaining components of the 3D Einstein equations are automatically satisfied. This shows that determining a near-horizon geometry is equivalent to solving \eqref{matrixeq} for the harmonic map scalars.

Finally, observe that for a near-horizon geometry the above calculations show that
\begin{equation}
(1-x^2)^2 \text{Tr}\left[\mathcal{M}^{-1}\partial_{x}\mathcal{M}
 \mathcal{M}^{-1}\partial_{x}\mathcal{M}\right] = 16.
\end{equation}
Hence $\mathfrak{I} = 0$ when evaluated at near-horizon geometries.

\section{Reparameterization of the Target and Area Lower Bound}	
\label{Sec5}

Suppose that $\mathcal{B}$ is diffeomorphic to $L(p,q)$ where $p$ and $q$ are mutually prime integers, and let $a^i_{+}\eta_{(i)}$ and $a^i_-\eta_{(i)}$ be the linear combinations of the $U(1)^2$ generators which vanish at $x=1$ and $x=-1$, respectively. This is equivalent to
\begin{equation}
a^i_+\lambda_{ij}=0 \text{ }\text{ at $x=1$},\qquad \qquad a^i_-\lambda_{ij}=0 \text{ } \text{ at $x=-1$},
\end{equation}
and without loss of generality \cite{Hollands2012} these direction vectors may be chosen to be
\begin{equation}
a_+=\begin{pmatrix}1\\0\end{pmatrix},\qquad a_-=\begin{pmatrix}{q}\\{p}\end{pmatrix}.
\end{equation}
In order to rewrite the area functional $\mathfrak{I}$ in a more convenient form,
first transform the lens direction vectors to that of the sphere. Namely set
\begin{equation}
Z=\begin{pmatrix}
1& q   \\
0 & p
\end{pmatrix},\qquad
Z^{-1}=\begin{pmatrix}
1& -\frac{q}{p}   \\
0 & \frac{1}{p}
\end{pmatrix},\qquad
\begin{pmatrix}
\phi^1   \\
\phi^2
\end{pmatrix}=Z\begin{pmatrix}
\bar{\phi}^1  \\
\bar{\phi}^2
\end{pmatrix},\qquad
\begin{pmatrix}
\bar{\phi}^1  \\
\bar{\phi}^2
\end{pmatrix}=Z^{-1}
\begin{pmatrix}
\phi^1   \\
\phi^2
\end{pmatrix},
\end{equation}
so that
\begin{equation}
{\lambda}_{ij}d\phi^id\phi^j=\underbrace{\left(Z^T{\lambda}Z\right)_{ij}}_{\bar{\lambda}_{ij}}
d\bar{\phi}^id\bar{\phi}^j,
\end{equation}
and
\begin{equation}
\bar{a}^i_+ \bar{\lambda}_{ij}=0 \text{ }\text{ at $x=1$},\qquad \qquad \bar{a}^i_- \bar{\lambda}_{ij}=0 \text{ } \text{ at $x=-1$},
\end{equation}
with
\begin{equation}
\bar{a}_+=\begin{pmatrix}1\\0\end{pmatrix},\qquad \bar{a}_-=\begin{pmatrix}{0}\\{1}\end{pmatrix}.
\end{equation}

Next select the following reparameterization of the $\lambda$ target space
variables
\begin{equation}
\bar{\lambda}_{11}=e^{2U+V}(1-x)\cosh W,\qquad \bar{\lambda}_{22}=e^{2U-V}(1+x)\cosh W,\qquad\bar{\lambda}_{12}=e^{2U}\sqrt{1-x^2}\sinh W,
\end{equation}
or rather
\begin{equation}
\lambda_{11}=e^{2U+V}(1-x)\cosh W,\quad \lambda_{12}=\frac{e^{2U}}{p}\left(\sqrt{1-x^2}\sinh W-q e^{V}(1-x)\cosh W\right),
\end{equation}
\begin{equation}
\lambda_{22}=\frac{e^{2U}}{p^2}\left(q^2e^{V}(1-x)\cosh W-2q\sqrt{1-x^2}\sinh W+e^{-V}(1+x)\cosh W\right),
\end{equation}
with inverse transformation
\begin{align}
\begin{split}
U&=\frac{1}{4}\log\left(\frac{\det\lambda}{p^2(1-x^2)}\right),\\	 V&=\frac{1}{2}\left(\frac{(1+x)\bar{\lambda}_{11}}{(1-x)\bar{\lambda}_{22}}\right)=\frac{1}{2}\left(\frac{p^2(1+x)\lambda_{11}}
{(1-x)\left[q^2\lambda_{11}-2q\lambda_{12}+\lambda_{22}\right]}\right),\\	 W&=\sinh^{-1}\left(\frac{\bar{\lambda}_{12}}
{e^{2U}\sqrt{1-x^2}}\right)
=\sinh^{-1}\left(\frac{\lambda_{12}-q\lambda_{11}}{pe^{2U}\sqrt{1-x^2}}\right).
\end{split}
\end{align}
Note that the regularity condition \eqref{eq2.14} becomes
\begin{equation}
\mathcal{C}^2=\lim_{x\to \pm 1} \frac{(1-x^2)^2}{\det {\lambda} \cdot a^i_\pm a^j_\pm {\lambda}_{ij}}=\lim_{x\to \pm 1} \frac{p^2(1-x^2)^2}{\det \bar{\lambda} \cdot \bar{a}^i_\pm \bar{a}^j_\pm \bar{\lambda}_{ij}},
\end{equation}
and therefore
\begin{equation}\label{conical singularity12}
\lim_{x\to \pm 1}e^{-6xU-V}=\frac{\mathcal{C}^2}{2p^2}.
\end{equation}
Moreover using \eqref{eq2.13} produces
\begin{equation}\label{regcondition}	
\xi(V+6xU)\Bigg|_{x=-1}^{x=1}=-(\xi(1)+\xi(-1))
\log\left(\frac{\mathcal{C}^2}{2p^2}\right)=-2\alpha_{\xi}\log\left(\frac{32\pi^4}{p^2A^2}\right),
\end{equation}
where
\begin{equation}
\alpha_{\xi}=\frac{\xi(1)+\xi(-1)}{2}.
\end{equation}

Let us now compute each term in $\mathfrak{I}$. Observe that
\begin{align}
\begin{split}
&(1-x^2)\left(\frac{1}{8} \frac{(\partial_x \det\lambda)^2}{(\det \lambda)^2} + \frac{\text{Tr} (\lambda^{-1} \partial_x\lambda)^2}{8}\right)-\frac{1}{1-x^2}\\
=&(1-x^2)\left(\frac{1}{8} \frac{(\partial_x \det \bar{\lambda})^2}{(\det \bar{\lambda})^2} + \frac{\text{Tr} (\bar{\lambda}^{-1} \partial_x\bar{\lambda})^2}{8}\right)-\frac{1}{1-x^2}\\	 =&\frac{1-x^2}{4}\left\{12\left(\partial_xU\right)^2+\left(\partial_xV\right)^2+(\partial_xW)^2+\sinh^2 W\left(\partial_xV+\partial_xh_2\right)^2\right\}\\
&-\frac{1}{2}\partial_xV-3x\partial_xU-\frac{3}{4},
\end{split}
\end{align}
where
\begin{equation}
h_2=\frac{1}{2}\log\left(\frac{1-x}{1+x}\right).
\end{equation}
Furthermore
\begin{align}
\begin{split}
&\frac{1}{4\det{\lambda}} \Theta^T_x\lambda^{-1}\Theta_x + \frac{1}{4 \det {\lambda}}\Upsilon^2_x + \frac{1}{4}\partial_x\psi^T\lambda^{-1}\partial_x\psi\\
=&\frac{1}{4\det{\lambda}} \Theta^T_x Z\bar{\lambda}^{-1}Z^T\Theta_x + \frac{1}{4 \det {\lambda}}\Upsilon^2_x + \frac{1}{4}\partial_x\psi^T Z\bar{\lambda}^{-1}Z^T\partial_x\psi\\
=&p^2\frac{e^{-6h_{1}-6U-h_{2}-V}}{\cosh W}(\bar{\Theta}^{1}_x)^{2}+p^2
e^{-6h_{1}-6U+h_{2}+V}\cosh W
\left(e^{-h_{2}-V}\tanh W\bar{\Theta}^{1}_x-\bar{\Theta}^{2}_x\right)^{2}\\
&
+p^2\frac{e^{-2h_{1}-2U-h_{2}-V}}{\cosh W}(\partial_x\bar{\psi}^{1})^{2}
+p^2e^{-2h_{1}-2U+h_{2}+V}\cosh W(e^{-h_{2}-V}\tanh W \partial_x\bar{\psi}^{1}-\partial_x\bar{\psi}^{2})^{2}\\
&+p^2e^{-4h_{1}-4U}\Upsilon_x^{2},
\end{split}
\end{align}
where
\begin{equation}
\bar{\Theta}_x=Z^T\Theta_x, \qquad \bar{\psi}=Z^T\psi, \qquad
h_1=\frac{1}{4}\log(1-x^2).
\end{equation}
Then setting $x=\cos\theta$, integrating by parts, and using the regularity condition \eqref{regcondition} produces
\begin{equation}
\mathfrak{I}=\mathcal{I}_{L(p,q)}
+\alpha_{\xi}\log\left(\frac{32\pi^4}{p^2A^2}\right)-\beta_{\xi}^0,
\end{equation}
where
\begin{equation}
\beta_{\xi}^{0}=\int_{0}^{\pi}\left(\frac{(\partial_{\theta}\xi)^2}{2\xi}+\frac{3}{4}\xi\right)\, \sin\theta d\theta
\end{equation}
and
\begin{align}
\begin{split} &\mathcal{I}_{L(p,q)}(\Psi)\\
=&\frac{1}{4}\int_{0}^{\pi}\xi\Bigg\{12(\partial_{\theta} U)^2+(\partial_{\theta} V)^2+(\partial_{\theta} W)^2+\sinh^2 W(\partial_{\theta} V+ \partial_{\theta} h_2)^2\\
&+ p^2\frac{e^{-6h_{1}-6U-h_{2}-V}}{\cosh W}(\bar{\Theta}_{\theta}^{1})^{2}+p^2
 e^{-6h_{1}-6U+h_{2}+V}\cosh W
\left(e^{-h_{2}-V}\tanh W\bar{\Theta}_{\theta}^{1}-\bar{\Theta}_{\theta}^{2}\right)^{2}\\
&+ p^2\frac{e^{-2h_{1}-2U-h_{2}-V}}{\cosh W}(\partial_{\theta}\bar{\psi}^{1})^{2}+ p^2 e^{-2h_{1}-2U+h_{2}+V}\cosh W(e^{-h_{2}-V}\tanh W \partial_{\theta}\bar{\psi}^{1}-\partial_{\theta}\bar{\psi}^{2})^{2}\\
&+ p^2 e^{-4h_{1}-4U}\Upsilon_{\theta}^{2}+3 p^2\Lambda\left(\frac{A}{2\sqrt{2}\pi^2}\right)^2e^{-4U}
-\xi^{-1}\left[2V\partial_{\theta}\xi+12U\partial_{\theta}
\left(\cos\theta\xi\right)\right]\partial_{\theta} h_2\Bigg\} \sin\theta d\theta,
\end{split}
\end{align}
with $\Psi=(U,V,W,\zeta^{1},\zeta^{2},\chi,\psi^{1},\psi^{2})$.

Consider now the case in which $\mathcal{B}$ is diffeomorphic to $S^1\times S^2$, that is when the same Killing vector $\eta_{(2)}$ vanishes at both $x=\pm 1$. Define the reparameterization
\begin{equation}
\lambda_{11}=e^{2U+\bar{V}}\cosh W,\qquad \lambda_{22}=e^{2U-\bar{V}}(1-x^2)\cosh W,\qquad \lambda_{12}=e^{2U}\sqrt{1-x^2}\sinh W,
\end{equation}
where
\begin{equation}
\bar{V}=V+2h_1+h_2.
\end{equation}
In a similar fashion to the computations above
\begin{align}
\begin{split}
&(1-x^2)\left(\frac{1}{8} \frac{(\partial_x \det \lambda)^2}{(\det \lambda)^2} + \frac{\text{Tr} (\lambda^{-1} \partial_x\lambda)^2}{8}\right)-\frac{1}{1-x^2}\\ =&\frac{1-x^2}{4}\left\{12\left(\partial_xU\right)^2+\left(\partial_x\bar{V}\right)^2
+(\partial_xW)^2+\sinh^2 W\left(\partial_x\bar{V}-2\partial_xh_1\right)^2\right\}
+\frac{x}{2}\partial_x\left(\bar{V}-6U\right)-1,
\end{split}
\end{align}
and
\begin{align}
\begin{split}
&\frac{1}{4\det{\lambda}} \Theta^T_x\lambda^{-1}\Theta_x + \frac{1}{4 \det {\lambda}}\Upsilon^2_x + \frac{1}{4}\partial_x\psi^T\lambda^{-1}\partial_x\psi\\
=&\frac{e^{-4h_{1}-6U-\bar{V}}}{\cosh W}(\Theta^{1}_{x})^{2}+
e^{-8h_{1}-6U+\bar{V}}\cosh W
\left(e^{2h_{1}-\bar{V}}\tanh W\Theta^{1}_x-\Theta^{2}_x\right)^{2}\\
&+\frac{e^{-2U-\bar{V}}}{\cosh W}(\partial_x\psi^{1})^{2}
+e^{-4h_{1}-2U+\bar{V}}\cosh W(e^{2h_{1}-\bar{V}}\tanh W \partial_x\psi^{1}-\partial_x\psi^{2})^{2}+e^{-4h_{1}-4U}\Upsilon_{x}^{2}.
\end{split}
\end{align}
Moreover in the current setting the lack of conical singularities yields
\begin{equation}
\lim_{x\to \pm 1}e^{-6U+\bar{V}}=\frac{\mathcal{C}^2}{2},
\end{equation}
so that
\begin{equation}	 x\xi(\bar{V}-6U)\Bigg|_{x=-1}^{x=1}=\left(\xi(1)+\xi(-1)\right)
\log\left(\frac{\mathcal{C}^2}{2}\right)=2\alpha_{\xi}\log\left(\frac{\mathcal{C}^2}{2}\right).
\end{equation}
Therefore
\begin{equation}
\mathfrak{I}=\mathcal{I}_{S^1\times S^2}+\alpha_{\xi}\log\left(\frac{32\pi^4}{A^2}\right)-\beta_{\xi}^{1},
\end{equation}
where
\begin{equation}
\beta_{\xi}^{1}=\int_{-1}^{1}\left((1-x^2)\frac{(\xi')^2}{2\xi}+\xi\right)\, dx
\end{equation}
and
\begin{align}
\begin{split} \mathcal{I}_{S^1\times S^2}=&\frac{1}{4}\int_{0}^{\pi}\xi\Bigg\{12\left(\partial_\theta U\right)^2+\left(\partial_{\theta}\bar{V}\right)^2+(\partial_{\theta}W)^2+\sinh^2 W\left(\partial_{\theta}\bar{V}-2\partial_{\theta}h_1\right)^2\\
&+\frac{e^{-4h_{1}-6U-\bar{V}}}{\cosh W}(\Theta^{1}_{\theta})^{2}+
e^{-8h_{1}-6U+\bar{V}}\cosh W
\left(e^{2h_{1}-\bar{V}}\tanh W\Theta^{1}_{\theta}-\Theta^{2}_{\theta}\right)^{2}\\
&+\frac{e^{-2U-\bar{V}}}{\cosh W}(\partial_{\theta}\psi^{1})^{2}
+e^{-4h_{1}-2U+\bar{V}}\cosh W(e^{2h_{1}-\bar{V}}\tanh W \partial_{\theta}\psi^{1}-\partial_{\theta}\psi^{2})^{2}\\	 &+e^{-4h_{1}-4U}\Upsilon_{\theta}^{2}+3\Lambda\left(\frac{A}{2\sqrt{2}\pi^2}\right)^2
e^{-4U}-\xi^{-1}\left(-2\bar{V}+12U\right)\partial_{\theta}(\cos\theta\xi)
\partial_{\theta}h_2\Bigg\}
\sin\theta d\theta.
\end{split}
\end{align}

It turns out that the two classes of functionals may be expressed in a unified fashion with the help of a parameter $s$, which takes the value 0 for the lens family of topologies and the value 1 for the topology $S^1\times S^2$. The relation
between topology and the values of $(p,q,s)$ is given by
\begin{equation}
\begin{cases}
\mathcal{B}\cong S^3,& s=0,\quad p=1,\quad q=0,\\
\mathcal{B}\cong L(p,q),& s=0,\quad 1\leq q\leq p-1,\\
\mathcal{B}\cong S^1\times S^2,& s=1,\quad p=1,\quad q=0.\\
\end{cases}
\end{equation}
Note that, in the ring case, the values of $p$ and $q$ do not coincide with those used earlier in the section. The purpose for using these values here is to unify the expression for the functional below. Let
\begin{equation}
V_s=V+2sh_1+sh_2,
\qquad \beta_{\xi}^s=\int_{0}^{\pi}\left(\frac{(\partial_{\theta}\xi)^2}{2\xi}
+\frac{3+s}{4}\xi\right)\, \sin\theta d\theta,
\end{equation}
then
\begin{equation}\label{RELATION}
\mathfrak{I}=\mathcal{I}_{\mathcal{B}}
+\alpha_{\xi}\log\left(\frac{32\pi^4}{p^2A^2}\right)-\beta_{\xi}^s,
\end{equation}
where
\begin{align}\label{AFunctional}
\begin{split}
&\mathcal{I}_{\mathcal{B}}(\Psi)\\
=&\frac{1}{4}\int_{0}^{\pi}\xi
\Bigg\{12\left(\partial_{\theta} U\right)^2+(\partial_{\theta} V_s)^2+(\partial_{\theta} W)^2+\sinh^2 W(\partial_{\theta} V+\partial_{\theta}  h_2)^2\\
&+p^2\frac{e^{-6h_{1}-h_2-6U-V}}{\cosh W}(\bar{\Theta}^{1}_{\theta})^{2}+
p^2e^{-6h_{1}+h_2-6U+V}\cosh W
\left(e^{-h_{2}-V}\tanh W\bar{\Theta}^{1}_{\theta}-\bar{\Theta}^{2}_{\theta}\right)^{2}\\
&+ p^2\frac{e^{-2h_1-h_2-2U-V}}{\cosh W}(\partial_{\theta}\bar{\psi}^{1})^{2}+p^2 e^{-2h_{1}+h_2-2U+V}\cosh W(e^{-h_2-V}\tanh W \partial_{\theta}\bar{\psi}^{1}-\partial_{\theta}\bar{\psi}^{2})^{2}\\
&+p^2 e^{-4h_{1}-4U}\Upsilon_{\theta}^{2}+3 p^2\Lambda\left(\frac{A}{2\sqrt{2}\pi^2}\right)^2e^{-4U}\\	 &-\xi^{-1}\left[2V_s\partial_{\theta}
\left(\left[1-2s\cos^2\frac{\theta}{2}\right]\xi\right)
+12U\partial_{\theta}\left(\cos\theta\xi\right)\right]\partial_{\theta}
h_2\Bigg\} \sin\theta d\theta.
\end{split}
\end{align}

\begin{proposition}\label{proposition3.2}
Let $(M, \mathbf{g}, F)$ be a bi-axisymmetric solution of 5-dimensional minimal supergravity, and let $\mathcal{B}$ be a bi-axisymmetric stable MOTS, then
\begin{equation}
A\geq \frac{4\sqrt{3}\pi^2}{p}e^{\frac{\mathcal{I}_{\mathcal{B}}(\Psi)
-\beta_{\xi}^s}{2\alpha_{\xi}}}.
\end{equation}
\end{proposition}

\begin{proof}
According to \eqref{eq3.13} the area functional satisfies $\mathfrak{I}\leq 0$.
The desired result then follows from \eqref{RELATION}.
\end{proof}

\section{Convexity of the Area Functional and Minimization}
\label{Sec6}

Consider the 3-sphere $S^3$ parameterized by Hopf coordinates $(\theta,\phi^1,\phi^2)$, where $\theta\in[0,\pi]$ and $\phi^i\in[0,2\pi]$, in which the round metric is expressed as
\begin{equation}
\frac{d\theta^2}{4}+\sin^{2}(\theta/2)(d\phi^1)^2+\cos^{2}(\theta/2)(d\phi^2)^2,
\end{equation}
with volume form
\begin{equation}
d\mathcal{V}=\frac{\sin\theta}{4}d\theta\wedge d\phi^1\wedge d\phi^2.
\end{equation}
Recall that the symmetric space $G_{2(2)}/SO(4)\cong\mathbb{R}^{8}$ comes equipped with a complete metric \cite{Alaee:2016jlp} of nonpositive curvature given by
\begin{align}\label{eq4.2}
\begin{split}
G=&12du^{2}
+\cosh^{2}w dv^{2}
+dw^{2}+p^2\frac{e^{-6u-v}}{\cosh w}(\bar{\Theta}^{1})^{2}
+p^2e^{-6u+v}\cosh w (e^{-v}\tanh w\bar{\Theta}^{1}-\bar{\Theta}^{2})^{2}\\
&+p^2\frac{e^{-2u-v}}{\cosh w}(d\bar{\psi}^{1})^{2}
+p^2e^{-2u+v}\cosh w(e^{-v}\tanh w d\bar{\psi}^{1}-d\bar{\psi}^{2})^{2}
+p^2e^{-4u}\Upsilon^{2}.
\end{split}
\end{align}
Let $\Omega\subset S^3$, then the quasi-harmonic energy on this domain
of maps $\tilde{\Psi}=(u,v,w,\zeta^{1},\zeta^{2},\chi,\bar{\psi}^{1},\bar{\psi}^{2}):
S^3\setminus\Gamma\rightarrow G_{2(2)}/SO(4)$, where $\Gamma$ is the union of the two circles $\theta=0,\pi$, is defined by
\begin{align}\label{eq4.1}
\begin{split}
E_{\Omega}(\tilde{\Psi})=&\frac{1}{\pi^2}\int_{\Omega}\xi\Bigg\{12(\partial_{\theta} u)^{2}
+\cosh^{2}w(\partial_{\theta} v)^{2}
+(\partial_{\theta}w)^{2}+p^2\frac{e^{-6u-v}}{\cosh w}(\bar{\Theta}^{1}_{\theta})^{2}\\
&
+p^2e^{-6u+v}\cosh w \left(e^{-v}\tanh w\bar{\Theta}_{\theta}^{1}-\bar{\Theta}_{\theta}^{2}\right)^{2}
+ p^2\frac{e^{-2u-v}}{\cosh w}(\partial_{\theta}\bar{\psi}^{1})^{2}\\
&+p^2 e^{-2u+v}\cosh w\left(e^{-v}\tanh w\partial_{\theta}\bar{\psi}^{1}-\partial_{\theta}\bar{\psi}^{2}\right)^{2}
+p^2 e^{-4u}\Upsilon_{\theta}^{2}\\
&+3 p^2\Lambda\left(\frac{A}{2\sqrt{2}\pi^2}\right)^2e^{-4u}\sin^2\theta\Bigg\} d \mathcal{V}.
\end{split}
\end{align}
This differs from the pure harmonic energy by the factor $\xi$ and the last term involving $\Lambda$. Next set $u=h_1+U$ and $v=h_2+V$, and observe that in Hopf coordinates
\begin{equation}
h_1=\frac{1}{2}\log\sin\theta,\quad\quad\quad h_2=\log\tan\frac{\theta}{2}.
\end{equation}
With the help of the identity
\begin{equation}
\partial_{\theta}h_1=\frac{\cos\theta}{2} \partial_{\theta}h_2,
\end{equation}
it follows that
\begin{equation}\label{eq4.4}
12\xi(\partial_{\theta} U)^2=12\xi(\partial_{\theta} u)^2-3\xi\cos^2\theta(\partial_{\theta} h_2)^2-12\partial_{\theta}\left(\xi U\cos\theta\right)\partial_{\theta} h_2+12U \partial_{\theta}(\xi\cos\theta) \partial_{\theta}h_2,
\end{equation}
and
\begin{equation}\label{eq4.5}
\begin{split}
\xi(\partial_{\theta}V_s)^2=&
\xi\left(\partial_{\theta}v+2s\partial_{\theta}h_1+(s-1)\partial_{\theta}h_2\right)^2\\
=&\xi(\partial_{\theta}v)^2
+\xi\left(\underbrace{2s\partial_{\theta}h_1+(s-1)\partial_{\theta}h_2}
_{\partial_{\theta}h_s}\right)^2
+2\xi\partial_{\theta}v\left(2s\partial_{\theta}h_1
+(s-1)\partial_{\theta}h_2\right)\\
=&\xi(\partial_{\theta}v)^2
+\xi(\partial_{\theta}h_s)^2
+2\xi\partial_{\theta}(V_s-h_s)\partial_{\theta}h_s\\
=&\xi(\partial_{\theta}v)^2-\xi(\partial_{\theta}h_s)^2
+2\xi\partial_{\theta}V_s\partial_{\theta}h_s\\
=&\xi(\partial_{\theta}v)^2-\xi\left(2s\cos^2\frac{\theta}{2}
-1\right)^2(\partial_{\theta}h_2)^2
+2\xi\left(2s\cos^2\frac{\theta}{2}-1\right)\partial_\theta V_s\partial_\theta h_2\\
=&\xi(\partial_{\theta}v)^2
-\xi\left(2s\cos^2\frac{\theta}{2}-1\right)^2(\partial_{\theta}h_2)^2
+2\partial_\theta\left(\xi\left(2s\cos^2\frac{\theta}{2}-1\right) V_s\right)\partial_\theta h_2\\
&-2V_s\partial_\theta\left(\xi\left(2s\cos^2\frac{\theta}{2}-1\right)\right) \partial_\theta h_2.
\end{split}
\end{equation}
Therefore, integration by parts and $\partial_{\theta}\left(\sin\theta\partial_{\theta}h_2\right)=0$ show that the area functional and quasi-harmonic energy are related by
\begin{equation}\label{eq4.6}
\begin{split}
4\mathcal{I}_{\Omega}(\Psi)=&E_{\Omega}(\tilde{\Psi})
-\int_{\Omega}\xi\left(\left(2s\cos^2\frac{\theta}{2}-1\right)^2
+3\cos^2\theta\right)(\partial_{\theta}h_2)^2d\mathcal{V}\\
&+\int_{\partial\Omega}\xi\left(2\left(2s\cos^2\frac{\theta}{2}-1\right)V_s-12\cos\theta U\right)\partial_{\nu}h_2 dA
\end{split}
\end{equation}
where $\nu$ is the unit outer normal and $\mathcal{I}_{\Omega}$ is the area functional \eqref{AFunctional} restricted to $\Omega$.

Let $\Psi_0=(U_0,V_0,W_0,\zeta_{0}^{1},\zeta_{0}^{2},\chi_0,\psi_{0}^1,\psi_{0}^2)$ be a renormalized quasi-harmonic map arising from the near-horizon geometry of the relevant model extreme black hole (mentioned in the statement of each theorem in Section \ref{Sec1}). In the appendix $\Psi_0$ is given explicitly, and it can be shown that $\Psi_0$ is a critical point of $\mathcal{I}_{\mathcal{B}}$. The goal of this section is to establish $\Psi_0$ as the global minimum point for $\mathcal{I}_{\mathcal{B}}$.

\begin{theorem}\label{minimum}
Suppose that $\Psi=(U,V,W,\zeta^{1},\zeta^{2},\chi,\psi^1,\psi^2)$ is smooth and satisfies the asymptotics \eqref{eq4.13}-\eqref{asdfghjkl}
with $\chi|_{\Gamma}=\chi_{0}|_{\Gamma}$, $\zeta^{i}|_{\Gamma}=\zeta^{i}_{0}|_{\Gamma}$, and $\psi^{i}|_{\Gamma}=\psi^{i}_{0}|_{\Gamma}$, $i=1,2$. Then there exists a constant $C>0$ such that
\begin{equation}\label{eq4.7}
\mathcal{I}_{\mathcal{B}}(\Psi)-\mathcal{I}_{\mathcal{B}}(\Psi_{0})
\geq C\int_{S^3}
\left(\operatorname{dist}_{G_{2(2)}/SO(4)}(\tilde{\Psi},\tilde{\Psi}_{0})-D\right)^2
d\mathcal{V},
\end{equation}
where $D$ denotes the average value of $\operatorname{dist}_{G_{2(2)}/SO(4)}(\tilde{\Psi},\tilde{\Psi}_{0})$.
\end{theorem}

The proof is based on a convexity argument. Namely, due to the fact that the target symmetric space $G_{2(2)}/SO(4)$ is nonpositively curved, and $\Lambda\geq 0$, the quasi-harmonic energy $E$ is convex under geodesic deformations. The functional $\mathcal{I}_{\mathcal{B}}$ then inherits such convexity as a result of \eqref{eq4.6}, which leads to the desired gap bound \eqref{eq4.7}. However, since the energy of the maps in question is infinite, a cut-and-paste argument away from the set $\Omega_{\varepsilon}=\{(\theta,\phi^1,\phi^2)\mid \sin\theta>\varepsilon\}$ is needed to apply the convexity property.

We first record all relevant asymptotic behavior. As $\theta\to 0,\pi$ the renormalized quasi-harmonic map satisfies
\begin{equation}\label{eq4.11}
U_0,\zeta^{1}_0,\zeta^{2}_0,\chi_0=O(1),\quad W_0=O(\sin\theta),\quad \partial_{\theta} U_0,\partial_{\theta}\chi_0,\partial_{\theta}\psi^i_0=O(\sin\theta),\quad \partial_{\theta} W_0= O(1),
\end{equation}
\begin{equation}\label{eq4.12}
V_0=\begin{cases}
O(1)& s=0\\
-2\log\left(\sin\frac{\theta}{2}\right)+O(1) &s=1
\end{cases},\quad \partial_{\theta}V_0=\begin{cases}
O(\sin\theta)& s=0\\
-\cot\frac{\theta}{2}+O(\sin\theta) &s=1
\end{cases},
\end{equation}
\begin{equation}
\psi^1_0=\begin{cases}
O(\sin^2\frac{\theta}{2}) & s=0\\
O(1) & s=1
\end{cases},\qquad \psi^2_0=\begin{cases}
O(\cos^2\frac{\theta}{2}) & s=0\\
O(1) & s=1
\end{cases},\qquad
\Theta^2_0=O(\sin^2\theta), \quad s=1,
\end{equation}
\begin{equation}
\partial_{\theta}\zeta^{1}_0=\begin{cases}
\sin^2\frac{\theta}{2}O(\sin\theta) & s=0\\
O(\sin\theta) & s=1
\end{cases},\qquad \partial_{\theta}\zeta^{2}_0=\begin{cases}
\cos^2\frac{\theta}{2}O(\sin\theta) & s=0\\
O(\sin\theta) & s=1
\end{cases}.
\end{equation}
Similarly the components of the given map $\Psi$ should satisfy
\begin{equation}\label{eq4.13}
U,\zeta^{1},\zeta^{2},\chi=O(1),\qquad W=O(\sin\theta),\qquad\partial_{\theta} W=O(1),\
\end{equation}
\begin{equation}\label{eq4.17}
V=\begin{cases}
O(1)& s=0\\
-2\log\left(\sin\frac{\theta}{2}\right)+O(1) &s=1
\end{cases},\qquad \partial_{\theta}V=\begin{cases}
O(\sin\theta)& s=0\\
-\cot\frac{\theta}{2}+O(\sin\theta) &s=1
\end{cases},
\end{equation}
\begin{equation}
\psi^1=\begin{cases}
O(\sqrt{\sin\frac{\theta}{2}}) & s=0\\
O(1) & s=1
\end{cases},\qquad \psi^2=\begin{cases}
O(\sqrt{\cos\frac{\theta}{2}}) & s=0\\
O(1) & s=1
\end{cases},\qquad
\Theta^2=O(\sin^2\theta), \quad s=1,
\end{equation}
\begin{equation}\label{eq4.15}
\partial_{\theta} U,\partial_{\theta}\chi=O({\sin\theta}),\quad \partial_{\theta}\psi^i=\begin{cases}
O(\sqrt{\sin{\theta}}) & s=0\\
O(\sin{\theta}) & s=1\\
\end{cases},
\end{equation}
\begin{equation}\label{asdfghjkl}
\partial_{\theta}\zeta^{1}=\begin{cases}
\sqrt{\sin\frac{\theta}{2}}O(\sin\theta) & s=0\\
O(\sin\theta) & s=1
\end{cases},\qquad \partial_{\theta}\zeta^{2}=\begin{cases}
\sqrt{\cos\frac{\theta}{2}}O(\sin\theta) & s=0\\
O(\sin\theta) & s=1
\end{cases}.
\end{equation}
Next, in order to carry out the cut-and-paste procedure define a Lipschitz cut-off function
\begin{equation}\label{eq4.16}
\varphi_{\varepsilon}=\begin{cases}
0 & \text{ if $\sin \theta\leq\varepsilon$,} \\
\frac{\log(\sin \theta/\varepsilon)}{\log(\sqrt{\varepsilon}/
\varepsilon)} &
\text{ if $\varepsilon<\sin \theta<\sqrt{\varepsilon}$,} \\
1 & \text{ if $\sin\theta\geq\sqrt{\varepsilon}$,} \\
\end{cases}
\end{equation}
and let
\begin{equation}\label{eq4.17+1}
\Psi_{\varepsilon}=(U,V_{\varepsilon},W_{\varepsilon},\zeta^{1}_{\varepsilon},
\zeta^{2}_{\varepsilon},\chi_{\varepsilon},\psi^1_{\varepsilon},\psi^2_{\varepsilon})
=(U,\Phi_{\varepsilon})=(U,\Phi_{0}+\varphi_{\varepsilon}\left(\Phi-\Phi_0\right))
\end{equation}
so that $\Psi_{\varepsilon}=(U,V_{0},W_{0},\zeta^{1}_{0},\zeta^{2}_{0},\chi_{0},
\psi^1_{0},\psi^2_{0})$ on $S^3 \setminus\Omega_{\varepsilon}$.

\begin{lemma}\label{lemmaepsilon}
$\lim_{\varepsilon\rightarrow 0}\mathcal{I}_{\mathcal{B}}
(\Psi_{\varepsilon})=\mathcal{I}_{\mathcal{B}}(\Psi)$.
\end{lemma}

\begin{proof}
Observe that
\begin{equation}\label{eq4.18}
\mathcal{I}_{\mathcal{B}}( \Psi_{\varepsilon})=\mathcal{I}_{\mathcal{B}}( \Psi_{\varepsilon})\vert_{\sin\theta\leq \varepsilon}+\mathcal{I}_{\mathcal{B}}( \Psi_{\varepsilon})\vert_{\varepsilon<\sin\theta<\sqrt{\varepsilon}}
+\mathcal{I}_{\mathcal{B}}( \Psi_{\varepsilon})\vert_{\sin\theta\geq\sqrt{\varepsilon}},
\end{equation}
and by dominated convergence theorem $\mathcal{I}_{\mathcal{B}}( \Psi_{\varepsilon})\vert_{\sin\theta\geq\sqrt{\varepsilon}}\to \mathcal{I}_{\mathcal{B}}(\Psi)$. Furthermore the first term on the right-hand side converges to zero since
\begin{align}\label{eq4.19}
\begin{split}
\mathcal{I}_{\mathcal{B}}( \Psi_{\varepsilon})\vert_{\sin\theta\leq \varepsilon}
=&\frac{1}{4\pi^2}\int_{\sin\theta\leq \varepsilon}\xi\Bigg\{\underbrace{12(\partial_{\theta} U)^{2}}_{O(\sin^2\theta)}
+\underbrace{(\partial_{\theta} V_0+2s\partial_{\theta}h_1+s\partial_{\theta}h_2)^{2}}_{O(\sin^2\theta)}
+\underbrace{(\partial_{\theta} W_{0})^{2}}_{O(\sin^2\theta)}  	 \\ &+\underbrace{\sinh^{2}W_{0}}_{O(\sin^2\theta)}\underbrace{(\partial_{\theta}
V_{0}+\partial_{\theta}h_{2})^{2}}_{O(\csc^2\theta)}
+ p^2\underbrace{e^{-2U-V_{0}}}_{\begin{cases}
O(1) & s=0\\
O(\sin^2\frac{\theta}{2})& s=1
\end{cases}}
\frac{\cot\frac{\theta}{2}}{\sin\theta\cosh W_{0}}\underbrace{|\partial_{\theta}\bar{\psi}^{1}_{0}|^{2}}_{O(\sin^2\theta)}
\\
&+p^2\underbrace{e^{-6U-V_{0}}}_{\begin{cases}
O(1) & s=0\\
O(\sin^2\frac{\theta}{2})& s=1
\end{cases}}
\frac{\cot\frac{\theta}{2}}{\sin^3\theta\cosh W_{0}}\underbrace{(\bar{\Theta}_{0}^{1})^{2}}_{\begin{cases}
O(\sin^2\theta\tan\frac{\theta}{2}) & s=0\\
O(\sin^2\theta)  & s=1
\end{cases}}\\
&+p^2
\underbrace{e^{-6U+V_{0}}}_{\begin{cases}
O(1) & s=0\\
O(\csc^{2}\frac{\theta}{2})& s=1
\end{cases}}\frac{\tan\frac{\theta}{2}\cosh W_{0}}{\sin^3\theta}
\underbrace{\left(\bar{\Theta}_{0}^{2}-e^{-V_{0}}\cot\frac{\theta}{2}\tanh W_{0}\bar{\Theta}_{0}^{1}\right)^{2}}_{O(\sin^2\theta)}
\Bigg\}d\mathcal{V}
\end{split}
\end{align}
\begin{align*}
\begin{split}
&+\frac{1}{4\pi^2}\int_{\sin\theta\leq \varepsilon}\xi\Bigg\{
p^2\underbrace{e^{-2U+V_{0}}}_{\begin{cases}
O(1) & s=0\\
O(\sin^{-2}\frac{\theta}{2})& s=1
\end{cases}}
\frac{\tan\frac{\theta}{2}\cosh W_{0}}{\sin \theta}\underbrace{\left(\partial_{\theta}\bar{\psi}_{0}^{2}
-e^{-V_{0}}\cot\frac{\theta}{2}\tanh W_{0}\partial_{\theta}\bar{\psi}_{0}^{1}\right)^{2}}_{O(\sin^2\theta)}\\
&+\frac{p^2}{\sin^2\theta}\underbrace{e^{-4U}}_{O(1)}	 \underbrace{|\Upsilon_{0}|^{2}}_{O(\sin^2\theta)}
+3p^2\Lambda\left(\frac{A}{2\sqrt{2}\pi^2}\right)^2\underbrace{ e^{-4U}}_{O(1)}\Bigg\}
d\mathcal{V}\\
&-\frac{1}{4\pi^2}\int_{\sin\theta\leq \varepsilon}\Bigg\{
\frac{12}{\sin\theta}\underbrace{U\partial_{\theta}\left(\cos\theta\xi\right)}_{O(1)}
+\frac{2}{\sin\theta}\underbrace{\left(V_0+2sh_1+sh_2\right)
\partial_{\theta}\left(\left[1-2s\cos^2\frac{\theta}{2}\right]\xi\right)}_{O(1)}\Bigg\}
d\mathcal{V}.
\end{split}
\end{align*}

Now consider the region $\mathcal{W}_{\varepsilon}=\{\theta\in[0,\pi]\,|\,\epsilon<\sin\theta<\sqrt{\epsilon}\}$. Let $I_{i}$, $i=1,\ldots,12$ denote integrals over $\mathcal{W}_{\varepsilon}$ of the terms in \eqref{eq4.19} with $\Phi_{0}$ replaced by $\Phi_{\varepsilon}$. Then
each such integral vanishes in the limit as $\varepsilon\rightarrow 0$. To see this observe that
\begin{equation}\label{eq4.23}
I_{1}\leq c\int_{\mathcal{W}_{\varepsilon}}
\underbrace{(\partial_{\theta} U)^{2}}_{O(\sin^2\theta)}d\mathcal{V}=O\left(\varepsilon\right),
\end{equation}
\begin{align}\label{eq4.24}
\begin{split}
I_{2}\leq & c\int_{\mathcal{W}_{\varepsilon}}\left(
\underbrace{(\partial_{\theta} V_0+2s\partial_{\theta}h_1+s\partial_{\theta}h_2)^{2}}_{O(\sin^2\theta)}
+\underbrace{(\partial_{\theta} V_{0}-\partial_{\theta} V)^{2}}_{
O(\sin^2\theta)}
+\underbrace{(V-V_{0})^{2}}_{O(1)}
\underbrace{(\partial_{\theta}\varphi_{\varepsilon})^{2}}_{ O\left(\cot^2\theta(\log\varepsilon)^{-2}\right)}\right)
d\mathcal{V}\\
=&O\left(\frac{1}{|\log{\varepsilon}|}\right),
\end{split}
\end{align}
\begin{equation}\label{eq4.25}
I_{3}\leq c\int_{\mathcal{W}_{\varepsilon}}
\left(\underbrace{(\partial_{\theta} W)^{2}}_{O(1)}+\underbrace{(\partial_{\theta}W_{0})^{2}}_{O(1)}
+\underbrace{(W-W_{0})^{2}}_{O(1)}\underbrace{(\partial_{\theta}\varphi_{\varepsilon})^{2}}_{ O\left(\cot^2\theta(\log\varepsilon)^{-2}\right)}\right)d\mathcal{V}
=O\left(\frac{1}{|\log{\varepsilon}|}\right).
\end{equation}	
Moreover, using $\sinh W_{\varepsilon}=O({\sin\theta})$ yields
\begin{align}\label{eq4.26}
\begin{split}
I_{4}\leq& c\int_{\mathcal{W}_{\varepsilon}}
\sin^2\theta\left(\underbrace{(\partial_{\theta} V-\partial_{\theta} V_{0})^{2}}_{O(\sin^2\theta)}
+\underbrace{(\partial_{\theta} h_{2})^{2}}_{O(\csc^2\theta)}	 +\underbrace{(V-V_{0})^{2}}_{O(1)}
\underbrace{(\partial_{\theta}\varphi_{\varepsilon})^{2}}_{ O\left(\cot^2\theta(\log\varepsilon)^{-2}\right)}\right)d\mathcal{V}\\
=&O(\varepsilon).
\end{split}
\end{align}
Next observe that since the values of the potentials of the two maps agree at the poles, it holds that
\begin{equation}\label{eq4.27}
|\zeta^i-\zeta^i_0|+|\chi-\chi_0|+|\bar{\psi}^i-\bar{\psi}^i_0|=O({\sin^2\theta}).
\end{equation}
From this we find
\begin{equation}\label{eq4.28}
|\partial_{\theta}\psi^{i}_{\varepsilon}|
\leq \underbrace{|\partial_{\theta}\psi^{i}|}_{\begin{cases}
O(\sqrt{\sin\theta}) & s=0\\
O({\sin\theta}) & s=1
\end{cases}}
+\underbrace{|\partial_{\theta}\psi_{0}^{i}|}_{ O(\sin\theta)}
+\underbrace{|\psi^{i}-\psi^{i}_{0}||\partial_{\theta}\varphi_{\varepsilon}|}
_{|\log\varepsilon|^{-1}O({\sin\theta})}
=\begin{cases}
O(\sqrt{\sin\theta})& s=0\\
O({\sin\theta})& s=1
\end{cases},
\end{equation}
and similar considerations produce
\begin{equation}\label{eq4.30}
|\Upsilon_{\varepsilon}|\leq
\begin{cases}	
|\log\varepsilon|^{-1} O(\sqrt{\sin\theta}) & s=0\\
 O(\sin\theta) & s=1
\end{cases},\qquad
|\bar{\Theta}_{\varepsilon}^{i}|\leq \begin{cases}
O(\sqrt{\sin\frac{\theta}{2}}{\sin\theta}) & s=0\\
|\log\varepsilon|^{1-i}O({\sin\theta}) & s=1
\end{cases}.
\end{equation}
It follows that
\begin{equation}\label{eq4.34}
I_{5}\leq c\int_{\mathcal{W}_{\varepsilon}} \frac{\cot\frac{\theta}{2}}{\sin\theta}\underbrace{\frac{e^{-V_{\varepsilon}-2U}}
{\cosh W_{\varepsilon}}}_{\begin{cases}
O(1) & s=0\\
O(\sin^2\frac{\theta}{2})& s=1
\end{cases}}\underbrace{(\partial_{\theta}\bar{\psi}_{\varepsilon}^{1})^{2}}_{
\begin{cases}
O({\sin\theta}) & s=0\\
O({\sin^2\theta}) & s=1
\end{cases}}d\mathcal{V},
\end{equation}
\begin{equation}\label{eq4.33}
I_{6}\leq c\int_{\mathcal{W}_{\varepsilon}}
\frac{\cot\frac{\theta}{2}}{\sin^3\theta}\underbrace{\frac{e^{-V_{\varepsilon}-6U}}
	{\cosh W_{\varepsilon}}}_{\begin{cases}
	O(1) & s=0\\
	O(\sin^{2}\frac{\theta}{2})& s=1
	 \end{cases}}\underbrace{(\bar{\Theta}_{\varepsilon}^{1})^{2}}_{
\begin{cases}
O(\sin\frac{\theta}{2}\sin^2\theta) & s=0\\
	O(\sin^2\theta) & s=1
	\end{cases}}d\mathcal{V},
\end{equation}
\begin{equation}\label{eq4.331}
I_{7}\leq c\int_{\mathcal{W}_{\varepsilon}}
\frac{\tan\frac{\theta}{2}}{\sin^3\theta}\underbrace{e^{-V_{\varepsilon}-6U}\cosh W_{\varepsilon}}_{\begin{cases}
	O(1) & s=0\\
	O(\csc^{2}\frac{\theta}{2})& s=1
	 \end{cases}}
\underbrace{(\bar{\Theta}^{2}_{\varepsilon}-e^{-V_{\varepsilon}}\cot\frac{\theta}{2} \tanh W_{\varepsilon}
	 \bar{\Theta}^{1}_{\varepsilon})^2}_{
\begin{cases}
O(\sin\frac{\theta}{2}\sin^2\theta) & s=0\\
	|\log\varepsilon|^{-2}O({\sin^2\theta}) & s=1
	\end{cases}}d\mathcal{V},
\end{equation}
\begin{equation}\label{eq4.35}
I_{9}\leq c\int_{\mathcal{W}_{\varepsilon}}
(\sin\theta)^{-2}\underbrace{e^{-4U}}_{O(1)}
\underbrace{|\Upsilon_{\varepsilon}|^{2}}_{\begin{cases}
|\log\varepsilon|^{-2} O({\sin\theta}),& s=0\\
 O(\sin^2\theta), & s=1
\end{cases}}d\mathcal{V}.
\end{equation}
Analogous estimates hold for the remaining integrals $I_{8}$, $I_{10}$, $I_{11}$, and $I_{12}$.
\end{proof}

A basic property of the harmonic energy, for maps into nonpositively curved target spaces, is convexity along geodesic deformations. When the cosmological constant is nonnegative this property carries over to the quasi-harmonic energy \eqref{eq4.1}.

\begin{proposition}\label{PropositionConvexity}
Let $\Omega\subset S^3$ be a domain which does not contain either of the poles $\theta=0,\pi$, and let $\tilde{\Psi}^t:\Omega\rightarrow G_{2(2)}/SO(4)$ be a family of smooth maps which are geodesics in $t\in [0,1]$. Then
\begin{equation}\label{eq4.38}
\frac{d^2}{dt^2}E_{\Omega}(\tilde{\Psi}^t)\geq \frac{1}{2\pi^2}
\int_{\Omega}|\nabla\mathrm{dist}_{G_{2(2)}/SO(4)}(\tilde{\Psi}^1,\tilde{\Psi}^0)|^2 d\mathcal{V}.
\end{equation}
\end{proposition}

\begin{proof}
The quasi-harmonic energy of the map $\tilde{\Psi}^t=(u_t,v_t,w_t,\zeta^{1}_{t},\zeta^{2}_{t},\chi_t,\psi^{1}_{t},
\psi^{2}_{t})$ is the sum of the pure harmonic energy scaled by $\xi$ and a term involving the cosmological constant, namely
\begin{equation}\label{eq4.39}
E_{\Omega}(\tilde{\Psi}^t)=\frac{1}{4\pi^2}\int_{\Omega}\xi|d\tilde{\Psi}^t|^2 d\mathcal{V}
+3p^2\Lambda\left(\frac{A}{2\sqrt{2}\pi^3}\right)^2\int_{\Omega}\xi e^{-4u_t}\sin^2\theta d\mathcal{V}
\end{equation}
where the energy density is
\begin{equation}\label{eq4.40}
|d\tilde{\Psi}^t|^2=4G_{BC}\partial_{\theta} (\tilde{\Psi}^t)^B
\partial_{\theta} (\tilde{\Psi}^t)^C.
\end{equation}
Since $G_{2(2)}/SO(4)$ is nonpositively curved the pure harmonic energy is convex along geodesics \cite{schoen2013convexity}, and the same is true for the scaling by $\xi$. In particular, using that $\xi\geq 1$ it holds that
\begin{equation}\label{eq4.41}
\frac{d^2}{dt^2}\int_{\Omega}\xi|d\tilde{\Psi}^t|^2 d\mathcal{V}\geq 2\int_{\Omega}|\nabla\mathrm{dist}_{G_{2(2)}/SO(4)}(\tilde{\Psi}^1,\tilde{\Psi}^0)|^2 d\mathcal{V}.
\end{equation}
Thus, it remains to show that
\begin{equation}\label{eq4.42}
\partial_{t}^{2}e^{-4u_t}=4(-\ddot{u}_t+4\dot{u}_t^2)e^{-4u_t}
=4\left[\Gamma^u_{BC}\partial_{t}(\tilde{\Psi}^t)^B \partial_{t}(\tilde{\Psi}^t)^C
+4\dot{u}_t^2\right]e^{-4u_t}\geq 0,
\end{equation}
where $\dot{u}_t=\partial_{t}u_t$ and the geodesic equation
\begin{equation}\label{eq4.43}
\ddot{u}_t+\Gamma^u_{BC}\partial_{t}(\tilde{\Psi}^t)^B \partial_{t}(\tilde{\Psi}^t)^C=0
\end{equation}
was used.

The Christoffel symbols may be computed as follows
\begin{equation}\label{eq4.44} \Gamma^{u}_{uu}=\Gamma^{u}_{uv}=\Gamma^{u}_{vv}=\Gamma^{u}_{uw}
=\Gamma^{u}_{u\zeta^i}=\Gamma^{u}_{u\psi^i}=\Gamma^{u}_{u\chi}=0,
\end{equation}
\begin{equation}\label{eq4.45} \Gamma^{u}_{vw}=\Gamma^{u}_{v\zeta^i}=\Gamma^{u}_{v\psi^i}=\Gamma^{u}_{v\chi}=
\Gamma^{u}_{ww}=\Gamma^{u}_{w\zeta^i}=\Gamma^{u}_{w\psi^i}=\Gamma^{u}_{w\chi}=0,
\end{equation}
\begin{equation}\label{eq4.47}
\Gamma^{u}_{\zeta^i\zeta^j}=\frac{1}{4}G_{\zeta^i\zeta^j},\quad \Gamma^{u}_{\zeta^i\bar{\psi}^j}=\frac{1}{4}G_{\zeta^i\bar{\psi}^j},\quad \Gamma^{u}_{\zeta^i\chi}=\frac{1}{4}G_{\zeta^i\chi},\quad
\Gamma^{u}_{\chi\chi}=\frac{1}{4}G_{\chi\chi}-\frac{p^2}{12}e^{-4u},
\end{equation}
\begin{equation}\label{eq4.48}
\Gamma^{u}_{\bar{\psi}^1\bar{\psi}^1}=\frac{1}{4}G_{\bar{\psi}^1\bar{\psi}^1}
-\frac{p^2}{6\cosh w}e^{-2u-v}-\frac{p^2}{6}e^{-2u-v}\sinh w\tanh w-\frac{p^2}{18}e^{-4u}(\bar{\psi}^2)^2,
\end{equation}
\begin{equation}\label{eq4.49}
\Gamma^{u}_{\bar{\psi}^1\bar{\psi}^2}=\frac{1}{4}G_{\bar{\psi}^1\bar{\psi}^2}
+\frac{p^2}{3}e^{-2u}\sinh w+\frac{p^2}{18}e^{-4u}\bar{\psi}^1\bar{\psi}^2,\quad
\Gamma^{u}_{\bar{\psi}^1\chi}=\frac{1}{4}G_{\bar{\psi}^1\chi}
+\frac{p^2}{6\sqrt{3}}e^{-4u}\bar{\psi}^2,
\end{equation}
\begin{equation}\label{eq4.51}
\Gamma^{u}_{\psi^2\chi}=\frac{1}{4}G_{\bar{\psi}^2\chi}
-\frac{p^2}{6\sqrt{3}}e^{-4u}\bar{\psi}^1,\quad
\Gamma^{u}_{\bar{\psi}^2\bar{\psi}^2}=\frac{1}{4}G_{\bar{\psi}^2\bar{\psi}^2}
-\frac{p^2}{36}e^{-4u}(\bar{\psi}^1)^2-\frac{p^2}{6}e^{-2u+v}\cosh w.
\end{equation}
Let
\begin{equation}\label{eq4.55}
\dot{\bar{\Theta}}_{t}^{i}=	\dot{\zeta}_{t}^i+ \psi_{t}^{i} \left( 	\dot{\chi}_{t} + \frac{1}{3\sqrt{3}}(\bar{\psi}_{t}^1 	 \dot{\bar{\psi}}_{t}^2 - \bar{\psi}_{t}^2 	 \dot{\bar{\psi}}^1) \right),\qquad \dot{\Upsilon}_{t}=	\dot{\chi}_{t} + \frac{1}{\sqrt{3}}(\bar{\psi}_{t}^1 	\dot{\bar{\psi}}_{t}^2 - \bar{\psi}_{t}^2	 \dot{\bar{\psi}}_{t}^1),
\end{equation}
and observe that
\begin{align}\label{eq4.56}
\begin{split}
|\partial_{t}\tilde{\Psi}^t|^2_{G}=&12\dot{u}_{t}^{2}
+\cosh^{2}w_{t} \dot{v}_{t}^{2}
+\dot{w}_{t}^{2}+\frac{p^2e^{-6u_{t}-v_{t}}}{\cosh w_{t}}(\dot{\bar{\Theta}}_{t}^{1})^{2}\\
&
+p^2e^{-6u_t+v_t}\cosh w_{t} (e^{-v_t}\tanh w_{t}\dot{\bar{\Theta}}_{t}^{1}-\dot{\bar{\Theta}}_{t}^{2})^{2}
+\frac{p^2e^{-2u_t-v_t}}{\cosh w_{t}}(\dot{\bar{\psi}}_{t}^{1})^{2}\\
&+p^2e^{-2u_t+v_t}\cosh w_{t}(e^{-v_t}\tanh w_{t} \dot{\bar{\psi}}_{t}^{1}-\dot{\bar{\psi}}_{t}^{2})^{2}
+p^2e^{-4u_{t}}\dot{\Upsilon}_{t}^{2}.
\end{split}
\end{align}
Therefore
\begin{align}\label{eq4.54}
\begin{split}	
\Gamma^u_{BC}\partial_{t}(\tilde{\Psi}^t)^B \partial_{t}(\tilde{\Psi}^t)^C
=&\frac{p^2e^{-6u_t-v_t}}{4\cosh w_{t}}(\dot{\bar{\Theta}}_{t}^{1})^{2}
+\frac{p^2}{4}e^{-6u_t+v_t}\cosh w_{t} (e^{-v_t}\tanh w_{t}\dot{\bar{\Theta}}_{t}^{1}-\dot{\bar{\Theta}}_{t}^{2})^{2}\\
&+\frac{p^2e^{-2u_t-v_t}}{12\cosh w_t}(\dot{\bar{\psi}}_{t}^{1})^{2}
+\frac{p^2}{12}e^{-2u_t+v_t}\cosh w_{t}(e^{-v_t}\tanh w_t \dot{\bar{\psi}}_{t}^{1}-\dot{\bar{\psi}}_{t}^{2})^{2}\\
&+\frac{p^2}{6}e^{-4u_t}\dot{\Upsilon}_{t}^{2},
\end{split}
\end{align}
confirming that \eqref{eq4.42} is nonegative.
\end{proof}

It is now possible to prove the main result of this section.

\begin{proof}[Proof of Theorem \ref{minimum}]
Let $\tilde{\Psi}^{t}_{\varepsilon}$, $t\in[0,1]$ be the minimizing geodesic in $G_{2(2)}/SO(4)$ connecting $\tilde{\Psi}_{0}$ to $\tilde{\Psi}_{\varepsilon}$. Then $U^{t}_{\varepsilon}=U_{0}+t(U-U_{0})$ and $V^{t}_{\varepsilon}=V_{0}$ on $S^3 \setminus \Omega_{\varepsilon}$. Observe that
\begin{equation}\label{eq4.57}
\frac{d^{2}}{dt^{2}}\mathcal{I}_{\mathcal{B}}(\Psi^{t}_{\varepsilon})
=\underbrace{\frac{d^{2}}{dt^{2}}\mathcal{I}_{\Omega_{\varepsilon}}
(\Psi^{t}_{\varepsilon})}_{I_{1}}+
\underbrace{\frac{d^{2}}{dt^{2}}\mathcal{I}_{S^3 \setminus \Omega_{\varepsilon}}
(\Psi^{t}_{\varepsilon})}_{I_{2}}.
\end{equation}
According to Proposition \ref{PropositionConvexity} and \eqref{eq4.6} we have
\begin{align}\label{eq4.58}
\begin{split}
I_{1}=&\frac{d^{2}}{dt^{2}}\frac{1}{4}E_{\Omega_{\varepsilon}}
(\tilde{\Psi}^{t}_{\varepsilon}) -\frac{d^{2}}{dt^{2}}\frac{1}{4}\int_{\Omega_{\varepsilon}}
\xi\left(\left(2s\cos^2\frac{\theta}{2}
-1\right)^2+3\cos^2\theta\right)(\partial_{\theta}h_2)^2 d\mathcal{V}\\	
&+\frac{d^{2}}{dt^{2}}\frac{1}{4}\int_{\partial\Omega_{\varepsilon}}
\xi\left(2\left(2s\cos^2\frac{\theta}{2}-1\right)(V_s)_0
-12\cos\theta \left(U_0+t(U-U_0)\right)\right)\partial_{\nu}h_2 dA\\
\geq & \frac{1}{8\pi^2}\int_{\Omega_{\varepsilon}}
|\nabla\operatorname{dist}_{G_{2(2)}/SO(4)}
(\tilde{\Psi}_{\varepsilon},\tilde{\Psi}_{0})|^{2}d\mathcal{V}.
\end{split}
\end{align}
Furthermore, using that $\operatorname{dist}_{G_{2(2)}/SO(4)}(\tilde{\Psi}_{\varepsilon},\tilde{\Psi}_{0})
=\sqrt{12}|u-u_{0}|$ on $S^3\setminus\Omega_{\varepsilon}$ yields
\begin{align}\label{eq4.59}
\begin{split}
I_{2}=&\frac{1}{4\pi^2}\int_{S^3\setminus \Omega_{\varepsilon}}
p^2\xi\Bigg\{\frac{24}{p^2}
(\partial_{\theta}U-\partial_{\theta}U_{0})^{2}
+36 (U-U_{0})^{2}
\frac{e^{-6h_1-h_2-6U^{t}-V_{0}}}{\cosh W_0}(\bar{\Theta}^{1}_{0})^{2}\\
&+36 (U-U_{0})^{2}
e^{-6h_1+h_2-6U^{t}+V_{0}}\cosh W_{0}(e^{-h_{2}-V_{0}}
\tanh W_{0}\bar{\Theta}^{1}_{0}-\bar{\Theta}_{0}^{2})^{2}\\
&+4(U-U_{0})^{2}
\frac{e^{-2h_1-h_2-2U^{t}-V_{0}}}{\cosh W_{0}}(\partial_{\theta}\bar{\psi}^{1}_{0})^2\\
&+4(U-U_{0})^{2}e^{-2h_1+h_2-2U^{t}+V_{0}}\cosh W_{0}
(e^{-h_2-V_{0}}\tanh W_{0}\partial_{\theta}\bar{\psi}_{0}^{1}-\partial_{\theta}\bar{\psi}_{0}^{2})^{2}\\
&+16(U-U_{0})^{2}e^{-4h_1-4U^t}(\Upsilon_{0})^2
+48 \Lambda\left(\frac{A}{2\sqrt{2}\pi^2}\right)^2(U-U_{0})^{2}e^{-4U^t}\Bigg\}
d\mathcal{V}\\
\geq &\frac{1}{2\pi^2}\int_{S^3\setminus\Omega_{\varepsilon}}
|\nabla\operatorname{dist}_{G_{2(2)}/SO(4)}
(\tilde{\Psi}_{\varepsilon},\tilde{\Psi}_{0})|^{2}d\mathcal{V}.
\end{split}
\end{align}
Note that passing $\frac{d^{2}}{dt^{2}}$ inside the integral is justified here
since all terms on the right-hand side of \eqref{eq4.59} are uniformly integrable.

Next using the fact that $\Psi_0$ is a critical point of the functional $\mathcal{I}_{\mathcal{B}}$, as well as the fact that in a neighborhood of the poles
\begin{equation}\label{eq4.62}
\frac{d}{dt}V^{t}_{\varepsilon}=
\frac{d}{dt}W^{t}_{\varepsilon}
=\frac{d}{dt}\zeta^{i,t}_{\varepsilon}		 =\frac{d}{dt}\chi^{t}_{\varepsilon}=\frac{d}{dt}\psi^{i,t}_{\varepsilon}
=0,\qquad i=1,2,
\end{equation}
shows
\begin{equation}\label{eq4.60}
\frac{d}{dt}\mathcal{I}_{\mathcal{B}}(\Psi^{t}_{\varepsilon})\bigg\vert_{t=0}	
=6\xi (U-U_{0})\partial_{\theta}U_{0}\sin\theta\bigg\vert^{\pi}_{0}=0.
\end{equation}
Now combine \eqref{eq4.57}-\eqref{eq4.59} and \eqref{eq4.60} to find
\begin{align}\label{eq4.64}
\begin{split}
\mathcal{I}_{\mathcal{B}}(\Psi_{\varepsilon})-\mathcal{I}_{\mathcal{B}}(\Psi_{0})
\geq& \frac{1}{8\pi^2}\int_{S^3}|\nabla\operatorname{dist}_{G_{2(2)}/SO(4)}
(\tilde{\Psi}_{\varepsilon},\tilde{\Psi}_{0})|^{2}d\mathcal{V}\\
\geq& C\int_{S^3}\left(\operatorname{dist}_{G_{2(2)}/SO(4)}	 (\tilde{\Psi}_{\varepsilon},\tilde{\Psi}_{0})-D_{\varepsilon}\right)^{2}d\mathcal{V},
\end{split}
\end{align}
where the second line arises from the Poincar\'e inequality and $D_{\varepsilon}$ is the average value of the distance between $\tilde{\Psi}_{\varepsilon}$ and $\tilde{\Psi}_{0}$. By Lemma \ref{lemmaepsilon}
$\lim_{\varepsilon\rightarrow 0}\mathcal{I}_{\mathcal{B}}
(\Psi_{\varepsilon})=\mathcal{I}_{\mathcal{B}}(\Psi)$, so the proof will be complete if it can be shown that this limit may be passed inside the integral. To accomplish this, observe that by the triangle inequality it is enough to verify
\begin{equation}\label{eq4.65}
\lim_{\varepsilon\rightarrow 0}\int_{S^3} \operatorname{dist}_{G_{2(2)}/SO(4)}^{2}(\tilde{\Psi}_{\varepsilon},\tilde{\Psi})
d\mathcal{V}=0.
\end{equation}
The triangle inequality implies
\begin{align}\label{eq4.66}
\begin{split} &\operatorname{dist}_{G_{2(2)}/SO(4)}(\tilde{\Psi}_{\varepsilon},\tilde{\Psi})\\
\leq&\operatorname{dist}_{G_{2(2)}/SO(4)}
((u,v_{\varepsilon},w_{\varepsilon},
\zeta^{1}_{\varepsilon},\zeta^{2}_{\varepsilon},\chi_{\varepsilon},
\psi^{1}_{\varepsilon},\psi^{2}_{\varepsilon}),
(u,v,w_{\varepsilon},
\zeta^{1}_{\varepsilon},\zeta^{2}_{\varepsilon},\chi_{\varepsilon},
\psi^{1}_{\varepsilon},\psi^{2}_{\varepsilon}))\\
&+\operatorname{dist}_{G_{2(2)}/SO(4)}
((u,v,w_{\varepsilon},
\zeta^{1}_{\varepsilon},\zeta^{2}_{\varepsilon},\chi_{\varepsilon},
\psi^{1}_{\varepsilon},\psi^{2}_{\varepsilon}),
(u,v,w,
\zeta^{1}_{\varepsilon},\zeta^{2}_{\varepsilon},\chi_{\varepsilon},
\psi^{1}_{\varepsilon},\psi^{2}_{\varepsilon}))\\
&+\cdots+
\operatorname{dist}_{G_{2(2)}/SO(4)}
((u,v,w,
\zeta^{1},\zeta^{2},\chi,
\psi^{1},\psi^{2}_{\varepsilon}),
(u,v,w,
\zeta^{1},\zeta^{2},\chi,
\psi^{1},\psi^{2}))\\
\leq& C\Bigg\{|v-v_{\varepsilon}|
+|w-w_{\varepsilon}|
+e^{-3u}\left(e^{-\tfrac{1}{2}v}
|\zeta^{1}-\zeta^{1}_{\varepsilon}|
+e^{\tfrac{1}{2}v}
|\zeta^{2}-\zeta^{2}_{\varepsilon}|\right)\\
&+ e^{-3u}\left(e^{-\tfrac{1}{2}v}
(|\psi^1|+|\psi^{1}_{0}|)
+e^{\tfrac{1}{2}v}
(|\psi^2|+|\psi^{2}_{0}|)\right)|\chi-\chi_{\varepsilon}|\\
&+e^{-3u}\left((|\psi^1|+|\psi^{1}_{0}|)
(|\psi^2|+|\psi^{2}_{0}|)e^{-\tfrac{1}{2}v}
+(|\psi^2|+|\psi^{2}_{0}|)^{2}e^{\tfrac{1}{2}v}\right)
|\psi^{1}-\psi^{1}_{\varepsilon}|\\
&+e^{-3u}\left((|\psi^1|+|\psi^{1}_{0}|)^{2}
e^{-\tfrac{1}{2}v}
+(|\psi^1|+|\psi^{1}_{0}|)(|\psi^2|+|\psi^{2}_{0}|)
e^{\tfrac{1}{2}v}\right)
|\psi^{2}-\psi^{2}_{\varepsilon}|\\
&+e^{-2u}\left(|\chi-\chi_{\varepsilon}|
+(|\psi^2|+|\psi^{2}_{0}|)|\psi^{1}-\psi^{1}_{\varepsilon}|	 +(|\psi^1|+|\psi^{1}_{0}|)|\psi^{2}-\psi^{2}_{\varepsilon}|\right)\\
&+e^{-u}\left(e^{-\tfrac{1}{2}v}|\psi^{1}-\psi^{1}_{\varepsilon}|
+e^{\tfrac{1}{2}v}|\psi^{2}-\psi^{2}_{\varepsilon}|\right)\Bigg\},
\end{split}
\end{align}
where it was used that distances between points of $G_{2(2)}/SO(4)$ are dominated by the length of connecting coordinate lines.
Since all terms on the right-hand side are uniformly bounded independent of $\varepsilon$, \eqref{eq4.65} follows from the dominated convergence theorem.
\end{proof}

\section{Proof of the Main Results}
\label{Sec7}

\textit{Proof of Theorems \ref{Thm1}, 2.2, \ref{Thm3}.}
From the spacetime $(M,\mathbf{g},F)$ and stable MOTS we obtain the map $\Psi$ as explained in Sections \ref{Sec4} and \ref{Sec5}.
Since $\Lambda=0$ it holds that $\alpha_{\xi}=1$, so that by Proposition \ref{proposition3.2}
\begin{equation}\label{00}
A\geq \frac{4\sqrt{3}\pi^2}{p}
e^{\frac{\mathcal{I}_{\mathcal{B}}(\Psi)-\beta^{s}_{1}}{2}}.
\end{equation}
Let $\Psi_0$ be the renormalized harmonic map arising from the near horizon geometry of the relevant model extreme black hole (mentioned in the statement of each theorem) having angular momentum and charges that agree with those of the given MOTS. Then according Theorem \ref{minimum}
\begin{equation}\label{000}
\frac{4\sqrt{3}\pi^2}{p}
e^{\frac{\mathcal{I}_{\mathcal{B}}(\Psi)-\beta^{s}_{1}}{2}}\geq
\frac{4\sqrt{3}\pi^2}{p}
e^{\frac{\mathcal{I}_{\mathcal{B}}(\Psi_0)-\beta^{s}_{1}}{2}}=A_0,
\end{equation}
where $A_0$ is the area of the horizon for the relevant model extreme black hole. Note that the equality in \eqref{000} follows from the fact that the function $\mathcal{M}$ vanishes
at $\Psi_0$, as is shown in Section \ref{Sec4}. In the appendices the value of $A_0$ is computed in terms of angular momentum and charges. This, together with \eqref{00} and \eqref{000} yields the desired area-angular momentum-charge inequality for each theorem.
In the case this inequality is saturated, we must have $\mathcal{I}_{\mathcal{B}}(\Psi)=\mathcal{I}_{\mathcal{B}}(\Psi_0)$ which implies by the gap bound that $\mathrm{dist}_{G_{2(2)}/SO(4)}(\tilde{\Psi},\tilde{\Psi}_0)$ is constant.
Since $\Psi$ realizes the infimum of the functional $\mathcal{I}_{\mathcal{B}}$, it is a critical point and hence a harmonic map. According to \cite{AlaeeKhuriKunduri_NHG} the two maps $\Psi$ and $\Psi_0$ must then be related by an isometry in the target symmetric space. The Maxwell field $F$ may then be reconstructed from $\Psi$ via \eqref{eq2.22}, and thus $(\mathcal{B},\gamma,F)$ must arise from the relevant near horizon geometry.
\hfill\qedsymbol\medskip

Before proceeding to the proof of Theorem \ref{thm5} we need a preliminary result.

\begin{lemma}\label{LMA}
\textit{(a)} Given $(A,\mathcal{J}_1,\mathcal{J}_2)\in \mathbb{R}^2_{+} \times \mathbb{R}_{-}$ with $\mathcal{J}=\mathcal{J}_1=\pm \mathcal{J}_{2}$, there exists a unique $(\hat{A},\hat{\mathcal{J}}_1,\hat{\mathcal{J}}_2)\in \mathbb{R}^2_{+} \times \mathbb{R}_{-}$ with $\hat{\mathcal{J}}=\hat{\mathcal{J}}_1=\pm \hat{\mathcal{J}}_2$ which saturates
\begin{equation}\label{AJin} \frac{\Lambda^3A^6}{2^{10}\pi^6\mathcal{J}^2}\leq
\frac{\left(A\sqrt{A^2+512\pi^2 \mathcal{J}^2}-A^2-128\pi^2\mathcal{J}^2\right)^3}
{\left(A-\sqrt{A^2+512 \pi^2 \mathcal{J}^2}\right)^4},
\end{equation}
and satisfies
\begin{equation}
	\hat{\mathcal{J}}=\frac{\mathcal{J}}{A^2}\hat{A}^2,\qquad \hat{A}\leq \frac{\pi^2}{\sqrt{2\Lambda^3}}.
\end{equation}	
Moreover, the inequality \eqref{AJin} is equivalent to $\hat{A}\geq A$.\smallskip

\textit{(b)} Given $(A,Q)\in \mathbb{R}^2_+$, there exists a unique $(\hat{A},\hat{Q})\in \mathbb{R}^2_+$ which saturates
\begin{equation}\label{AQin}
Q^2\leq \frac{12}{64\pi^2}\left(\frac{A\pi}{2}\right)^{4/3}-\frac{3\Lambda A^2}{32\pi^2},
\end{equation}
and satisfies
\begin{equation}
\hat{Q}=\frac{Q}{A}\hat{A},\qquad \hat{A}\leq \frac{\pi^2}{\sqrt{2\Lambda^3}}.
\end{equation}	
Moreover, the inequality \eqref{AQin} is equivalent to $\hat{A}\geq A$.
\end{lemma}

\begin{proof}
Consider part (a) when $\mathcal{J}_{1}=-\mathcal{J}_2$; similar arguments hold when $\mathcal{J}_1=\mathcal{J}_2$. We may assume without loss of generality that $\mathcal{J}_1>0$. Define the curve
\begin{equation}
		f(\tau)=(A(\tau),\mathcal{J}_1(\tau),\mathcal{J}_2(\tau))=\left(\tau, \frac{\mathcal{J}_1}{A^2}\tau^2,\frac{\mathcal{J}_2}{A^2}\tau^2\right)=\left(\tau, \frac{\mathcal{J}}{A^2}\tau^2,-\frac{\mathcal{J}}{A^2}\tau^2\right)
\end{equation}
in $\mathbb{R}^2_{+} \times \mathbb{R}_{-}$. Then for small $\tau$ each side of the inequality satisfies
\begin{equation}
		\frac{\Lambda^3A^6(\tau)}{2^{10}\pi^6\mathcal{J}^2(\tau)}\sim\tau^{2},\qquad \frac{\left(A(\tau)\sqrt{A^2(\tau)+512\pi^2 \mathcal{J}^2(\tau)}
-A^2(\tau)-128\pi^2\mathcal{J}^2(\tau)\right)^3}{\left(A(\tau)-\sqrt{A^2(\tau)
+512\pi^2 \mathcal{J}^2(\tau)}\right)^4} \sim 1,
\end{equation}
so that the inequality holds on the curve $f$. For large $\tau$ it is
clear that the inequality is reversed. It follows that there exists a time $\tau= \hat{A}$ for which the inequality is saturated. Further analysis of the roots of the associated polynomial show that this time is unique for $\tau\leq \frac{\pi^2}{\sqrt{2\Lambda^3}}$.

In order to establish the last statement in part (a), we interpret $\mathbb{R}^2_{+}$ as having a vertical $\mathcal{J}$-axis and horizontal $A$-axis. Observe that inequality \eqref{AJin}
corresponds to all points lying below the surface defined by equality in \eqref{AJin}; this is similar to
Figure 1 in Appendix A. According to the description of $\hat{A}$ above, it follows that $A\leq\hat{A}$ if and only if the inequality \eqref{AJin} is satisfied.

Similar arguments may be used to establish part (b) with the curve
\begin{equation}
f(\tau)=(A(\tau),Q(\tau))=\left(\tau, \frac{Q}{A}\tau\right)
\end{equation}
in $\mathbb{R}^2_+$.	
\end{proof}

\textit{Proof of Theorem \ref{thm5}.}
We will provide details only for part (a), as similar arguments may be used for part (b). Let $\Psi$ be the map obtained from the spacetime $(M,\mathbf{g},F)$ as explained in Sections \ref{Sec4} and \ref{Sec5}, and let $(A,\mathcal{J}_1,\mathcal{J}_2)$ be the area and angular momenta of the stable MOTS. Lemma \ref{LMA} states that there exist corresponding values $(\hat{A},
\hat{\mathcal{J}}_1,\hat{\mathcal{J}}_2)$ which arise from an extreme
CCLP black hole, and are such that the desired inequality \eqref{mainJin} is equivalent to showing $\hat{A}\geq A$. Let $\hat{a}$ and $\hat{b}$ be the angular momentum parameters for this extreme CCLP solution (these quantities are given implicitly in terms of $\hat{\mathcal{J}}_1$ and $\hat{\mathcal{J}}_2$ by the equations \eqref{eqA.7}), and set
\begin{align}
\begin{split}
\hat{\Psi}=&(\hat{U},\hat{V},\hat{W},\hat{\zeta}^1,\hat{\zeta}^2,
\hat{\chi},\hat{\psi}^1,
\hat{\psi}^2)\\
=&\left(U+\frac{1}{2}\log\frac{\hat{A}}{A},V,W,\left(\frac{\hat{A}}{A}\right)^{3/2}\zeta^1,
\left(\frac{\hat{A}}{A}\right)^{3/2}\zeta^2,
\frac{\hat{A}}{A}\chi,\left(\frac{\hat{A}}{A}\right)^{1/2}\psi^1,
\left(\frac{\hat{A}}{A}\right)^{1/2}\psi^2\right).
\end{split}
\end{align}
According to Theorem \ref{minimum}
\begin{equation}\label{asdf0}
\hat{\mathcal{I}}_{S^{3}}(\hat{\Psi})\geq\hat{\mathcal{I}}_{S^3}(\hat{\Psi}_0),
\end{equation}
where $\hat{\Psi}_0$ denotes the extreme CCLP map with the same angular momenta $\hat{\mathcal{J}}_1$, $\hat{\mathcal{J}}_2$, and $\hat{\mathcal{I}}_{S^3}$ represents the functional
$\mathcal{I}_{S^3}$ defined with respect to the quantities $\hat{a}$, $\hat{b}$, and $\hat{A}$.
Next observe that
\begin{equation}
\hat{\mathcal{I}}_{S^3}(\Psi)=\hat{\mathcal{I}}_{S^3}(\hat{\Psi})
-3\alpha_{\hat{\xi}}\log\frac{\hat{A}}{A},
\end{equation}
and therefore with the help of Proposition \ref{proposition3.2}
\begin{equation}\label{asdf}
A\geq 4\sqrt{3}\pi^2 e^{\frac{\hat{\mathcal{I}}_{S^3}(\Psi)-\hat{\beta}^{0}_{\hat{\xi}}}{2\alpha_{\hat{\xi}}}}
=\frac{4\sqrt{3}\pi^2 A^{3/2}}{\hat{A}^{3/2}} e^{\frac{\hat{\mathcal{I}}_{S^3}(\hat{\Psi})-\hat{\beta}^{0}_{\hat{\xi}}}{2\alpha_{\hat{\xi}}}}.
\end{equation}
By combining \eqref{asdf0} and \eqref{asdf} we obtain at $\hat{A}\geq A$, since
\begin{equation}
4\sqrt{3}\pi^2 e^{\frac{\mathcal{I}_{S^3}(\hat{\Psi}_0)-\hat{\beta}^{0}_{\hat{\xi}}}{2\alpha_{\hat{\xi}}}}
=\hat{A}.
\end{equation}

Consider now the case of equality in \eqref{mainJin}. By the proof of Lemma \ref{LMA} this implies that $(\hat{A},\hat{\mathcal{J}}_1,\hat{\mathcal{J}}_2)=(A,\mathcal{J}_1,\mathcal{J}_2)$, and
hence $\hat{\Psi}=\Psi$, $\hat{\Psi}_0=\Psi_0$. Furthermore $\mathcal{I}_{S^3}(\Psi)=\mathcal{I}_{S^3}(\Psi_0)$, which as in the above proof of Theorems \ref{Thm1}, 2.2, and \ref{Thm3} yields $\Psi=\Psi_0$ up to isometry in the target symmetric space. From here the same arguments apply to show that $(\mathcal{B},\gamma,F)$ must arise from the near-horizon geometry of the extreme CCLP black hole.
\hfill\qedsymbol\medskip

\appendix

\section{The CCLP Charged Rotating de Sitter Black Hole}

\subsection{The Solution}

Consider 5-dimensional minimal supergravity with a positive cosmological constant with action \eqref{sugraaction}. The Chong-Cvetic-Lu-Pope (CCLP) solution \cite{chong2007non} may be interpreted as the natural generalization of the Kerr-Newman de Sitter black hole to 5 dimensions. In Boyer-Lindquist coordinates the solution takes the form
\begin{align}\label{eqA.1}
	\begin{split}
		\mathbf{g}=&-\frac{\xi\left[\left(1-\Lambda r^2\right)\Sigma dt+2q\nu\right]dt}{\Xi_a\Xi_b\Sigma}+\frac{2q\nu\omega}{\Sigma}+\frac{f}{\Sigma^2}\left(\frac{\xi dt}{\Xi_a\Xi_b}-\omega\right)^2\\
		&+\frac{\Sigma dr^2}{\Delta}+\frac{\Sigma d\tilde\theta^2}{\xi}+\frac{r^2+a^2}{\Xi_a}\sin^2\tilde\theta (d\phi^1)^2+\frac{r^2+b^2}{\Xi_b}\cos^2\tilde\theta (d\phi^2)^2
	\end{split}
\end{align}
where
\begin{equation}\label{eqA.2}
	\nu=b\sin^2\tilde\theta d\phi^1+a\cos^2\tilde\theta d\phi^2,\qquad \omega=a\sin^2\tilde\theta\frac{d\phi^1}{\Xi_a}+b\cos^2\tilde\theta\frac{d\phi^2}{\Xi_b},
\end{equation}
\begin{equation}\label{eqA.3}
	\xi=1+\Lambda(a^2 \cos^2\tilde\theta+b^2\sin^2\tilde\theta),\quad \Delta=\frac{(r^2+a^2)(r^2+b^2)(1-\Lambda r^2)+q^2+2abq}{r^2}-2\mathfrak{m},
\end{equation}
\begin{equation}\label{eqA.4}
	\Sigma={r}^2+b^2\sin^2\tilde\theta+a^2\cos^2\tilde\theta,\qquad \Xi_a=1+a^2\Lambda,\qquad \Xi_b=1+b^2\Lambda,\qquad f=(2\mathfrak{m}-2abq\Lambda)\Sigma-q^2.
\end{equation}
The Maxwell field $F = d\mathcal{A}$ has the potential
\begin{equation}\label{eqA.5} \mathcal{A}=\frac{\sqrt{3}q}{\Sigma}\left(\frac{\xi dt}{\Xi_a\Xi_b}-\omega\right).
\end{equation}
The above solution is characterized by the parameters $(\mathfrak{m},a,b,q)$ and for a suitable range of parameters they describe regular black holes on and outside an event horizon up to a cosmological horizon whose scale is set by the length scale $\ell^2 = \Lambda^{-1}$.  If $\Lambda=0$ the solution is considerably simpler as we discuss below.  Here $t \in \mathbb{R}$, $\tilde\theta \in (0,\pi/2)$, and $\phi^i \sim \phi^i + 2\pi$. Apart from  coordinate singularities at $\tilde\theta =0, \pi/2$ where the rotational Killing fields $\partial/\partial \phi^i$ degenerate, there are singularities at the roots of $\Delta(r)$. These correspond to an inner horizon, an outer horizon, and a cosmological horizon for suitable choice of parameters.  In particular for the subfamily of extreme black holes we require the cubic function $\Delta(R)$ with $R \equiv r^2$ to have three real positive roots, two of which coincide.  This condition is equivalent to requiring the discriminant of the cubic $\Delta(R)$ vanishes, which reduces to an equation of the form $f(a,b,q,\mathfrak{m}) =0$ for a smooth function $f$. The implicit function theorem will guarantee that, generically in some open set in parameter space, a solution $\mathfrak{m} = \mathfrak{m}(a,b,q)$ exists.

If we set $R_+$ to be a root we can eliminate $\mathfrak{m}$ by
\begin{equation}
\mathfrak{m} = \frac{(R_+ + a^2)(R_++b^2)(1- R_+ \Lambda) + q^2 + 2 a b q}{2R_+},
\end{equation}
and $R_+$ is a double root provided
\begin{equation}\label{extremecond}
\Lambda R_+^2 (2R_+ + a^2 + b^2) = R_+^2 - (a b + q)^2,
\end{equation}
which implies $R_+ \geq | a b + q|$ with equality if and only if $\Lambda =0$.  We will take $q\geq 0$ (below we will see this is equivalent to choosing electric charge $Q\geq 0$) and assume $ab + q > 0$.  Finally we require that the cosmological horizon $R_c \geq R_+$, that is
\begin{equation}
 (a^2 + b^2) (a b + q)^2 + 3 (a b + q)^2 R_+ - R_+^3 \geq 0.
\end{equation}
In summary, the extreme family is parameterized by $(R_+, a, b, q)$ which satisfy the extremality constraint \eqref{extremecond}.

Defining $r_+ = \sqrt{ R_+}$, we can derive the quasi-harmonic map data corresponding to the near-horizon geometry associated to the extreme subfamily of black holes. The horizon metric is
\begin{equation}\label{eqA.8}
	\gamma_{mn} dy^m dy^n=\frac{\Sigma_{r_+}}{\xi}d\tilde\theta^2+\lambda_{ij}d\phi^i d\phi^j,
\end{equation}
where $\Sigma_{r_+}={r}_+^2+b^2\sin^2\tilde\theta+a^2\cos^2\tilde\theta$ and
\begin{equation}\label{eqA.9} \lambda_{11}=\frac{(r_+^2+a^2)\sin^2\tilde\theta}{\Xi_a}+\frac{a\left[a\left(2\mathfrak{m}\Sigma_{r_+}-q^2\right)+2bq \Sigma_{r_+}\right]\sin^4\tilde\theta}{\Sigma_{r_+}^2\Xi_a^2},
\end{equation}
\begin{equation}\label{eqA.10} \lambda_{22}=\frac{(r_+^2+b^2)\cos^2\tilde\theta}{\Xi_b}
+\frac{b\left[b\left(2\mathfrak{m}\Sigma_{r_+}
-q^2\right)+2aq\Sigma_{r_+}\right]\cos^4\tilde\theta}{\Sigma_{r_+}^2\Xi_b^2},
\end{equation}
\begin{equation}\label{eqA.11}
\lambda_{12}=\frac{\left[ab\left(2\mathfrak{m}\Sigma_{r_+}-q^2\right)+(a^2+ b^2)q\Sigma_{r_+}\right]\sin^2\tilde\theta\cos^2\tilde\theta}{\Sigma_{r_+}^2\Xi_b\Xi_a},
\end{equation}
and the horizon area is $A = 8\pi^2 \mathcal{C}^{-1}$ where
\begin{equation}
\mathcal{C}^{-1}  = \frac{r_+^4 + r_+^2(a^2 + b^2) + a b (ab + q)}{4r_+ \Xi_a \Xi_b}.
\end{equation}
It is straightforward to read off the magnetic potentials
\begin{equation}\label{eqA.17}
	\psi_0^1=\frac{\sqrt{3}aq\sin^2\tilde\theta}{\Xi_a\Sigma_{r_+}},\qquad \psi_0^2=\frac{\sqrt{3}bq\cos^2\tilde\theta}{\Xi_b\Sigma_{r_+}}.
\end{equation}
A somewhat longer calculation gives $\Upsilon$, and hence the potential
\begin{equation}\label{eqA.22}
\chi_0=-\frac{\sqrt{3}q(b^2+r^2_+)(a^2+r^2_+)}{\Xi_a\Xi_b(a^2-b^2)\Sigma_{r_+}}.
\end{equation}
The computation of the charged twist potentials $\zeta_0^i$ is involved and yields cumbersome expressions which we will omit here, although we will record the asymptotic behavior relevant for the convexity argument below.  Using the quasi-harmonic map potentials one can calculate the electric charge
\begin{equation}\label{eqA.6}
	Q=\frac{\sqrt{3}\pi q}{4\Xi_a\Xi_b}   \; ,
\end{equation}
and angular momenta associated to the black hole horizon
\begin{equation}\label{eqA.7} \mathcal{J}_{1}=\frac{\pi\left[2a\mathfrak{m}+qb\left(2-\Xi_a\right)\right]}{4\Xi_a^2\Xi_b},\qquad \mathcal{J}_{2}=\frac{\pi\left[2b\mathfrak{m}+qa\left(2-\Xi_b\right)\right]}{4\Xi_a\Xi_b^2}  \; .
\end{equation}

The estimates of the quasi-harmonic map as $\tilde\theta\to 0,\pi/2$ will now be collected. Recall that in terms of our parametrization for spherical topology, we have
\begin{equation}\label{eqA.13}	 U_0=\frac{1}{4}\log\left(\frac{\xi(r^4_++(a^2+b^2)r^2_++a^2b^2+abq)^2}{4\Sigma_{r_+}r^2_+\Xi_a^2\Xi_b^2}\right),\qquad V_0=\frac{1}{2}\log\left(\frac{\lambda_{11}\cos^2\frac{\tilde\theta}{2}}{\lambda_{22}\sin^2\frac{\tilde\theta}{2}}\right).
\end{equation}
The required asymptotic behavior of the scalars is
\begin{equation}\label{eqA.26} V_0=\frac{1}{2}\log\left(\frac{4\pi^2(a^2+r^2_+)^3}{A^2\Xi^3_a}\right)+O(\sin^22\tilde\theta),\qquad \partial_{\tilde\theta} V_0=O(\sin2\tilde\theta),
\end{equation}
\begin{equation}\label{eqA.27} U_0=\frac{1}{4}\log\left(\frac{\Xi_aA^2}{16\pi^4(a^2+r_+^2)}\right)+O(\sin^22\tilde\theta),\qquad \partial_{\tilde\theta} U_0=O(\sin2\tilde\theta),
\end{equation}
\begin{equation}\label{eqA.28}
	W_0=O(\sin2\tilde\theta),\qquad\partial_{\tilde\theta} W_0=O(1),
\end{equation}
\begin{equation}\label{eqA.29}
	\psi^1_0= O(\sin^2\tilde\theta),\quad {\partial_{\tilde\theta}\psi^1_0}=O(\sin\tilde\theta),\qquad \psi^2_0=O(\cos^2\tilde\theta),\quad {\partial_{\tilde\theta}\psi^2_0}=O(\sin\tilde\theta),
\end{equation}
\begin{equation}\label{eqA.30}
	\zeta^1_0=-\frac{2\mathcal{J}_1}{\pi}+O(\sin^4\tilde\theta)\quad\text{ as $\tilde\theta\to 0$},\qquad \zeta^1_0=\frac{2\mathcal{J}_1}{\pi}+O(\cos^2\tilde\theta)\quad\text{ as $\tilde\theta\to \pi/2$},
\end{equation}
\begin{equation}\label{eqA.31}
	\zeta^2_0=-\frac{2\mathcal{J}_2}{\pi}+O(\sin^2\tilde\theta)\quad\text{ as $\tilde\theta\to 0$},\qquad \zeta^2_0=\frac{2\mathcal{J}_2}{\pi}+O(\cos^4\tilde\theta)\quad\text{ as $\tilde\theta\to \pi/2$},
\end{equation}
\begin{equation}\label{eqA.32}
	\chi_0=\frac{2Q}{\pi}+O(\sin^2\tilde\theta)\quad\text{ as $\tilde\theta\to 0$},\qquad \chi_0=-\frac{2Q}{\pi}+O(\cos^2\tilde\theta)\quad\text{ as $\tilde\theta\to \pi/2$},
\end{equation}
\begin{equation}\label{eqA.33}
	\partial_{\tilde\theta}\chi_0=O(\sin 2\tilde\theta ),\qquad {\partial_{\tilde\theta}\zeta^1_0}=\sin2\tilde\theta O(\sin^2\tilde\theta),\qquad {\partial_{\tilde\theta}\zeta^2_0}= \sin2\tilde\theta O(\cos^2\tilde\theta).
\end{equation}

\subsection{Geometric Equalities Satisfied by the CCLP Black Hole Horizon}

In this section we consider some special subfamilies of the three-parameter family of extreme horizons.

\subsubsection{Vanishing Cosmological Constant $\Lambda =0 $}

In this case the geometry simplifies significantly as we can express $R_+$ explicitly in terms of the parameters as $r^2_+ = R_+ = ab + q>0$, or equivalently
\begin{equation}
\mathfrak{m} = q + \frac{(a+b)^2}{2}.
\end{equation}
The extreme horizon area satisfies
\begin{equation}\label{eqA1.30}
A_\text{e}= 8\sqrt{\pi^2\mathcal{J}_1\mathcal{J}_2+\frac{4\pi}{3\sqrt{3}}Q^3},
\end{equation}
and positivity is guaranteed by $a b+ q > 0$.
Note that if set $Q=0$, we recover the vacuum result obtained by Hollands \cite{Hollands2012}. Moreover, if either $\mathcal{J}_i$ is set to zero, then we obtain
\begin{equation}\label{eqA1.32}
A_{\text{e}} =  8\sqrt{\frac{4\pi}{3\sqrt{3}}Q^3}.
\end{equation}
This holds in particular for an extreme Reissner-Nordstr\"{o}m horizon.

\subsubsection{Vanishing Angular Momenta $\mathcal{J}_i=0$}  It is sufficient (although not necessary) to set $ a= b=0$.  In this case the extremality condition \eqref{extremecond} reads
\begin{equation}\label{RNdScond}
 2 \Lambda R_+^3 = R_+^2 - q^2,
 \end{equation} and the mass parameter is fixed by
\begin{equation}
\mathfrak{m} = \frac{R_+^2 + 3q^2}{4R_+}.
\end{equation}
The geometry of the horizon is that of a round $S^3$, with area $A = 2\pi^2 R_+^{3/2}$ and using \eqref{RNdScond} gives the geometric relation
\begin{equation}
6\left(\frac{\pi A_\text{e}}{2}\right)^{4/3} - 3\Lambda A_\text{e}^2 = 32 \pi^2 Q^2,
\end{equation}
or equivalently
\begin{equation}
\mathcal{E}(A_\text{e},Q)=A_\text{e}^2\left(4\pi^2-2\Lambda\left(4\pi A_\text{e}\right)^{2/3}\right)^{3/2}-\frac{2^{11}\pi^4Q^3}{3\sqrt{3}} = 0.
\end{equation}
It can be seen that
\begin{equation}
0 \leq A_{\text{e}} \leq A_{\text{max}} = \frac{\pi^2}{\sqrt{2\Lambda^3}} \;, \qquad  0 \leq Q \leq Q_{\text{max}} = \frac{\pi}{12 \Lambda} \;.
\end{equation}
$A_{\text{max}}$ is achieved when $Q=0$ and the cosmological and event horizon coincide.
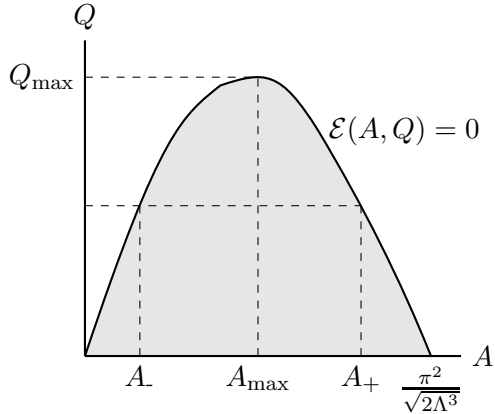
\begin{figure}[h]
	\centering
	\begin{tikzpicture}
	\draw[fill=gray!20!white,thick] (0,0) .. controls (1,3) and (1.3,3.2) ..(1.8,3.6).. controls (2.4,3.8) and (2.6,3.8) .. (3.1,3)node[right]{$\mathcal{E}(A,Q)=0$}.. controls (3.7,2) and (4.2,1)..(4.6,0)node[below]{$\frac{\pi^2}{\sqrt{2\Lambda^{3}}}$};
	\draw[thick] (0,0) -- (5,0) node[right] {$A$};
	\draw[thick] (0,0) -- (0,4.2) node[above] {$Q$};
	\draw[dashed] (2.3,0)node[below] {$A_\text{max}$} -- (2.3,3.71);
	\draw[dashed] (0,3.71)node[left] {$Q_\text{max}$} -- (2.3,3.71);
	\draw[dashed] (0,2)-- (3.67,2);
	\draw[dashed] (.73,0)node[below] {$A_\text{-}$} -- (.73,2);
	\draw[dashed] (3.67,0)node[below] {$A_\text{+}$} -- (3.67,2);
	\end{tikzpicture}
	\caption{The shaded region represents the region $\mathcal{E}(A,Q)\geq 0$.}
\end{figure}

\subsubsection{Equal Angular Momenta $\mathcal{J}=\mathcal{J}_1=\pm \mathcal{J}_2$ and Vanishing Electric Charge $Q=0$}

Consider the one-parameter subset defined by $a = b >0$ and $q =0$.  Then from \eqref{extremecond} we find
\begin{equation}
a^2 = R_+(1 - 2R_+ \Lambda)\;, \qquad \mathfrak{m} = 2R_+(1 - R_+ \Lambda)^3,
\end{equation}
so that $R_+ < (2\Lambda)^{-1}$.  The area of the extreme horizon is given by
\begin{equation}
A_{\text{e}} = \frac{8\pi^2 R_+^{3/2}}{(1 + 2R_+ \Lambda)^2}.
\end{equation}
This implies the following complicated relation between $A$ and $\mathcal{J}$
\begin{equation}
\frac{\Lambda^3A_\text{e}^6}{2^{10}\pi^6\mathcal{J}^2}=
\frac{\left(A_\text{e}\sqrt{A_\text{e}^2+512\mathcal{J}^2\pi^2}
-A_\text{e}^2-128\pi^2\mathcal{J}^2\right)^3}{\left(A_\text{e}
-\sqrt{A_\text{e}^2+512\mathcal{J}^2\pi^2}\right)^4}.
\end{equation}
From the above bound on $R_+$ it follows that
\begin{equation}
0\leq A_{\text{e}} \leq \frac{\pi^2}{\sqrt{2\Lambda^{3}}},\qquad 0\leq \mathcal{J}\leq  \mathcal{J}_{\text{max}} \equiv \frac{\sqrt{2}\pi}{54\Lambda^{3/2}}.
\end{equation}
In the case $\mathcal{J}=0$ the maximal area occurs when the event horizon and cosmological horizon coincide, while the maximum angular momentum is achieved when $A_{\text{e}} =\frac{4\pi^2}{9\Lambda^{3/2}}$.

\section{Dipole Charged Black Rings}\label{app2}

An explicit 3-parameter family of asymptotically flat stationary bi-axisymmetric black ring solutions, characterized by a mass $\mathfrak{m}$, vanishing electric charge $Q=0$,  a single angular momenta $\mathcal{J}_1$ along the $S^1$-direction of the black ring, and a `dipole charge' $\mathcal{D}$ which corresponds to the flux of the Maxwell field out of the $S^2$ portion of the ring, was constructed in \cite{Emparan:2004wy}. In the physics literature this class of solutions is referred to as the `singly-spinning dipole ring' for this reason. If $\mathcal{D} =0$, then the solution reduces to the vacuum black ring with one angular momentum \cite{emparan2002rotating}.  A remarkable feature of this solution is that it demonstrates `continuous non-uniqueness' - that is, for fixed $\mathfrak{m}, \mathcal{J}_1$, there are an infinite number of distinct dipole rings.  When $\mathcal{D} \neq 0$, the dipole ring admits a two-parameter extreme limit.  The associated near-horizon geometry, given in \cite{Kunduri:2007vf}, corresponds to an extreme horizon with $S^1 \times S^2$ topology. Note that from the point of view of the near-horizon alone, the radius of the $S^1$ is a free parameter although for the parent asymptotically flat black hole the radius is fixed.  Accordingly, we will leave it here as a free parameter $R_1$. The harmonic map scalars can be read off from the horizon metric
\begin{equation}\label{eqD.1}
\frac{dx^2}{\mathcal{C}^2 \det \lambda} + \lambda_{ij} d\phi^i d\phi^j =\frac{dx^2}{\mathcal{C}^2 \det \lambda} +  \frac{R_1^2 \sigma (1+\sigma) H(x)}{\mu (1-\sigma) F(x)} (d\phi^1)^2 + \frac{R_2^2 \mu^2 \omega_0^2 (1-x^2)}{ H(x)^2} (d\phi^2)^2,
\end{equation}
where
\begin{equation}\label{eqD.2}
F(x) = 1 + \sigma x, \qquad H(x)  = 1 - \mu x,
\end{equation}
and $\sigma, \mu \in (0,1)$.  The local metric extends smoothly to a metric on $S^1 \times S^2$ provided conical singularities are removed, which requires
\begin{equation}\label{eqD.3}
\omega_0 = \sqrt{F(1) H(1)^3} = \sqrt{F(-1) H(-1)^3}.
\end{equation}
This imposes a constraint on the parameters $\sigma, \mu$, given by
\begin{equation}\label{eqD.4}
(1+\sigma)(1-\mu)^3 = (1-\sigma)(1+\mu)^3,
\end{equation}
which can actually be solved explicitly
\begin{equation}\label{eqD.5}
\sigma = \frac{\mu(3+\mu^2)}{1+3\mu^2}
\end{equation}
so that
\begin{equation}\label{eqD.6}
\omega_0^2 = \frac{(1-\mu^2)^3}{1+3\mu^2}.
\end{equation}
The solution is parameterized by $(R_1, R_2, \sigma, \mu)$ subject to \eqref{eqD.5}.  We have also made the identification
\begin{equation}\label{eqD.7}
\mathcal{C}^{-1} = L^3 \equiv \omega_0 R_1 R_2^2 \sqrt{\frac{\sigma(1+\sigma) \mu^3}{1-\sigma}}.
\end{equation}
The remaining scalars are
\begin{equation}\label{eqD.8}
\chi_0 \equiv 0, \qquad \zeta_0^1 = \frac{L^3 R_1 (1+\sigma)}{R_2\sqrt{\sigma \mu^3}}\left(1 - \frac{1}{F(x)}\right), \qquad \psi_0^2 = -\sqrt{3} \sqrt{\frac{1 - \mu}{1+\mu}} \cdot \frac{\omega_0 \mu R_2(1+x)}{H(x)},
\end{equation}
and $\zeta_0^2 =\psi_0^1 =0$.

By the definition of electric charge and $\chi_0=0$, we have $Q=0$.  The dipole charge is
\begin{equation}\label{eqD.9}
\mathcal{D} = -v^i (\psi_0^i(+1) - \psi_0^i(-1)) = -\psi_0^2(+1) = \frac{2\sqrt{3} \mu R_2 (1-\mu^2)}{\sqrt{1+3\mu^2}},
\end{equation}
where $v^i = (0,1)$ corresponds to the fact that $\eta_2$ is the Killing vector field which vanishes at the poles of the $S^2$ of the ring horizon.  The angular momenta can be derived from the twist charged potentials, giving
\begin{equation}\label{eqD.11}
\mathcal{J}_1 = \frac{\pi}{2} \frac{L^3 R_1}{(1-\sigma) R_2}\sqrt{\frac{\sigma}{\mu^3}} \;, \qquad \mathcal{J}_2=0 \;.
\end{equation}
The area of the extreme horizon is
\begin{equation}\label{eqD.12}
A_e= 8\pi^2 L^3 = 4\pi \sqrt{ \frac{\pi \mathcal{J}_1 D^3}{3\sqrt{3}}}.
\end{equation}
Note that there is no limit as $\mathcal{D} \to 0$ or $\mathcal{J}_1 \to 0$; that is, the extreme dipole ring requires both a non-vanshing angular momenta along the $S^1$ direction of the ring, and a non-vanishing dipole charge. Therefore the area inequality is
\begin{equation}\label{eqD.121}
A\geq 4\pi \sqrt{ \frac{\pi \mathcal{J}_1 D^3}{3\sqrt{3}}}.
\end{equation}

Lastly, we set $x=\cos\theta$ and list the asymptotics of the harmonic map as $\theta\to 0,\pi$:
\begin{equation}\label{eqD.13}
	V_0=-\log\left(\sin\frac{\theta}{2}\right)+O(1),\qquad \partial_{\theta} V_0=-\cot\frac{\theta}{2}+O(\sin\theta),
\end{equation}
\begin{equation}\label{eqD.14}
	 U_0=\frac{1}{4}\log\left(\frac{\mu R_2^2(1+ \mu)^3(3 + \mu^2)R_1^2}{(3\mu^2+1)(1-\mu)}\right)+O(\sin^2\theta),\qquad \partial_{\theta} U_0= O(\sin\theta),\qquad W_0=0,
\end{equation}
\begin{equation}\label{eqD.16}
 \psi^2_0= \frac{\sqrt{12}(1-\mu^2)\mu R_2}{\sqrt{3\mu^2+1}}\cos^2\frac{\theta}{2}+O(\sin^2\theta),\qquad {\partial_{\theta}\psi^2_0}= O(\sin\theta),
\end{equation}
\begin{equation}\label{eqD.17}
	\zeta^1_0=-\frac{2\mathcal{J}_1}{\pi}+O(\sin^2\theta)\quad\text{ as $\theta\to 0$},\qquad \zeta^1_0=\frac{2\mathcal{J}_1}{\pi}+O(\sin^2\theta)\quad\text{ as $\theta\to \pi$},\qquad \zeta^2_0=0,
\end{equation}
\begin{equation}\label{eqD.19}
	\chi_0=0,\qquad {\partial_{\theta}\zeta^1_0}=O(\sin\theta).
\end{equation}

\section{A Magnetically Charged Kerr String}\label{app3}

Consider the class of extreme horizons with topology $S^1 \times S^2$ obtained as follows. Start with a general vacuum Kerr black hole solution.  Add a flat direction to obtain the product metric
\begin{equation}\label{eqE.1}
\mathbf{g} = \mathbf{g}_{\textrm{Kerr}} + dz^2,
\end{equation}
which is obviously Ricci flat in $D=5$, and hence one has a vacuum solution with horizon topology $S^1 \times S^2$ where $z$ is periodically identified with period $2\pi R$. The resulting 3-parameter vacuum solution is referred to as a `Kerr black string'. Note that the solution obtained is \emph{not} asymptotically flat but rather is asymptotically $\mathbb{R}^{3,1} \times S^1$.  Solution generating techniques, based on the underlying harmonic map structure of the theory,  can be used to generate solutions to minimal supergravity with electric and magnetic charge (as measured from the $D=4$ point of view) as well as linear momentum  along the string direction $z$.   Taking an extreme limit of this charged Kerr string and performing the near-horizon limit, we obtain extreme horizons with horizon topology $S^1 \times S^2$. Remarkably, the near-horizon geometry of the asymptotically flat vacuum black ring is globally isometric to a subfamily of the near-horizon geometries of the vacuum Kerr black string \cite{Kunduri:2007vf}. This strongly suggests there could be some (yet to be explicitly  constructed) family of extreme charged black rings with horizon geometry globally isometric to that of a charged Kerr black string.

In the following we will consider an extreme horizon parameterized by $(a,\beta, R)$, and we use the shorthand $c_\beta = \cosh \beta, s_\beta = \sinh \beta$. The Killing part of the metric is given by
\begin{align}\label{eqE.2}
\begin{split}
\lambda_{ij} d\phi^i d\phi^j =&  \frac{a^4(1-x^2)}{ \Xi(x)} \left(2(c_\beta^4 + s_\beta^4) d\phi^2 \right)^2 \\
&+ \left[ R d\phi^1 + \frac{2a^3 c_\beta s_\beta (c^2_\beta + s^2_\beta) (1-x^2)}{\Xi(x)} d\phi^2 \right]^2,
\end{split}
\end{align}
where
\begin{equation}\label{eqE.3}
\Xi(x) = a^2(1 + x^2 + 4c_\beta^2 s_\beta^2).
\end{equation}
The horizon scale is set by
\begin{equation}\label{eqE.4}
\mathcal{C}^{-1} = L^3 =  2 a^2 R (c_\beta^4 + s_\beta^4).
\end{equation}
It is easily seen that $\partial/ \partial \phi^2$ has fixed points at $x = \pm 1$, and the above metric extends smoothly to a cohomogeneity-one metric on $S^1 \times S^2$.  The remaining scalars are
\begin{equation}\label{eqE.5}
\psi_0^1 = 0,
\qquad \psi_0^2 = -\frac{4\sqrt{3} a^3 s_\beta c_\beta(c_\beta^4 + s_\beta^4)x}{\Xi(x)},\qquad \chi_0 = \frac{2\sqrt{3} a L^3 c_\beta s_\beta (c_\beta^2 + s_\beta^2)}{\Xi(x)}.
\end{equation}
A computation gives
\begin{eqnarray}
\Theta_0^1 &=& -\frac{2 L^3 R a c_\beta s_\beta (1 - x^2 + 4c_\beta^2 s_\beta^2)}{ \Xi(x)^2}  \; dx,\\
\Theta_0^2 &=& -\frac{4L^3 ( 1+ 2 s_\beta^2)(1 + 7 s_\beta2 + 19 s_\beta^4 + 24 s_\beta^6 12 s_\beta^8 + c_\beta^2 x^2 + s_\beta^4 x^2)}{(1 + 4c_\beta^2 s_\beta^2 + x^2)^3}(1-x^2) \; dx,
\end{eqnarray}
with twist potentials expressed concisely as
\begin{equation}\label{eqE.7}
\zeta_0^1 = -\frac{2L^3 R a c_\beta s_\beta x}{\Xi(x)},  \qquad \zeta_0^2 = -\frac{4L^3 a^2(c_\beta^2 + s_\beta^2)(1+ s^2_\beta c_\beta^2)x}{\Xi(x)} - \frac{\chi\psi_0^2}{3}.
\end{equation}
Owing to the functional form of $\chi_0$, $\psi_0^i$ it can be verified that $Q =0$, so the solution has vanishing electric charge.  There is a dipole charge
\begin{equation}\label{eqE.8}
\mathcal{D} = 4\sqrt{3} a c_\beta s_\beta
\end{equation}  as well as two angular momenta given by
\begin{equation}
\mathcal{J}_1 = -\pi a R^2 c_\beta s_\beta \;, \qquad \mathcal{J}_2 = - 2\pi a ^2 R (c_\beta^2 + s_\beta^2)\;.
\end{equation} One can verify that the following equality holds for the extreme solution
\begin{equation}\label{eqE.10}
A_{\text{e}}= 8\pi^2 \mathcal{C}^{-1} = 8\pi \sqrt{\mathcal{J}_2^2 - \frac{\pi}{12\sqrt{3}} \mathcal{J}_1 \mathcal{D}^3}.
\end{equation}
Then the area inequality is
\begin{equation}\label{eqE.101}
A\geq  8\pi \sqrt{\mathcal{J}_2^2 - \frac{\pi}{12\sqrt{3}} \mathcal{J}_1 \mathcal{D}^3}.
\end{equation}

Lastly, we set $x=\cos\theta$ and list the asymptotics of the harmonic map as $\theta\to 0,\pi$:
\begin{equation}\label{eqE.11}
 V_0=-\log\left(\sin\frac{\theta}{2}\right)+O(1),\qquad \partial_{\theta} V_0=-\cot\frac{\theta}{2}+O(\sin\theta),\qquad W_0=O(\sin\theta),
\end{equation}
\begin{equation}\label{eqE.12}
 U_0=\frac{1}{4}\log\left(2R^2a^2\right)+O(\sin^2\theta),\qquad\partial_{\theta} U_0= O(\sin\theta), \qquad\partial_{\theta} W_0=O(1),
\end{equation}
\begin{equation}\label{eqE.13}
 \psi^1_0=0,\qquad  \psi^2_0=-\frac{\mathcal{D}}{2}+O(\sin^2{\theta})\quad\text{ as $\theta\to 0$},\quad \psi^2_0=\frac{\mathcal{D}}{2}+O(\sin^2{\theta})\quad\text{ as $\theta\to\pi$},
\end{equation}
\begin{equation}\label{eqE.14}
 \chi_0=2\sqrt{3}aRs_{\beta}c_{\beta}(2c_{\beta}^2-1)+ O(\sin^2\theta),\qquad \partial_{\theta}\chi_0=O(\sin\theta)\quad\text{ as $\theta\to 0,\pi$},
\end{equation}
\begin{equation}\label{eqE.15}
 \zeta^1_0=-\frac{2\mathcal{J}_1}{\pi}+O(\sin^2\theta)\quad\text{ as $\theta\to 0$},\qquad \zeta^1_0=\frac{2\mathcal{J}_1}{\pi}+O(\sin^2\theta)\quad\text{ as $\theta\to \pi$},
\end{equation}
\begin{equation}\label{eqE.16}
 \zeta^2_0=-\frac{2\mathcal{J}_2}{\pi}+O(\sin^2\theta)\quad\text{ as $\theta\to 0$},\qquad \zeta^2_0=\frac{2\mathcal{J}_2}{\pi}+O(\sin^2\theta)\quad\text{ as $\theta\to \pi$},
\end{equation}
\begin{equation}\label{eqE.17}
\partial_{\theta}\psi^2_0,{\partial_{\theta}\zeta^1_0},\, {\partial_{\theta}\zeta^2_0}=O(\sin\theta),
\end{equation}
\begin{equation}
 \Theta^1 = O(\sin\theta) d \theta,  \qquad \Theta^2 = O(\sin^2\theta) d \theta \quad\text{ as $\theta\to 0, \pi$}.
\end{equation}

\section{Lens Space $L(n,1)$ Horizons}\label{app4}

The supergravity theory discussed here admits special classes of \emph{supersymmetric} (BPS) solutions. These are solutions which admit Killing spinors with respect to an appropriate connection.  Within the class of BPS solutions, there are asymptotically flat black hole solutions that must saturate the bound $M = \sqrt{3}Q$ \cite{Gibbons:1993xt}.  A BPS black hole is necessarily extreme (see e.g. \cite{Kunduri:2008tk}), so they  immediately give rise to near-horizon geometries that will be critical points of our harmonic map equations.

There is a classification  (without any isometry assumptions) of possible BPS near-horizon geometry solutions $(g,F)$ \cite{Reall:2002bh}. The only possibilities are: (i) $S^1 \times S^2$ with a product metric and the $S^2$ metric is round; (ii)  $S^3$ with a homogeneously squashed $SU(2)
\times U(1)$ metric or a lens space quotient thereof. There is also a $T^3$ possibility with the flat metric, but this is ignored because such a horizon could not correspond to an asymptotically flat black hole \cite{galloway2006rigidity,Galloway2006}. Case (i) is realized by the family of BPS black ring solutions \cite{Elvang:2004rt}.  The $S^3$ possibility is realized by the asymptotically flat black hole solutions \cite{Breckenridge:1996is, Kunduri:2014iga}, and more recently asymptotically flat black holes with lens horizon $L(n,1)$ have been constructed in \cite{kunduri2014supersymmetric,Tomizawa:2016kjh}.


The near-horizon geometries of case (ii) above are all locally isometric, with
\begin{eqnarray}
\mathbf{g}_{NH} &=& -\frac{r^2 dv^2}{\alpha^2} + \frac{4 \alpha dv dr}{j} dv dr + \frac{4\beta}{n\alpha^2} r dv (d\phi^1 + \frac{n}{2}\cos\theta d\phi^2) + ds^2_3, \\
F_{NH} &=& \frac{\sqrt{3}}{\alpha} d \left[ r dv - \frac{2\beta}{n}(d\phi^1 + \frac{n}{2}\cos\theta d\phi^2)\right],
\end{eqnarray}
where $j = \sqrt{2\alpha^3 - \beta^2}$.  The orientation is such that $\epsilon_{v r \theta \phi^1 \phi^2} >0$. The solution has two continuous parameters $\alpha , \beta$ satisfying the regularity conditions $\alpha >0$, $2\alpha^3 - \beta^2 >0$. The angles $\phi^i$ both have period $2\pi$ and $\theta \in (0,\pi)$.  The positive integer $n$ labels the horizon metric
\begin{equation}
\gamma_{ml} dy^m dy^l = \frac{4 j^2}{n^2 \alpha^2}\left(d\phi^1 + \frac{n}{2} \cos\theta d\phi^2\right)^2 + 2\alpha(d\theta^2 + \sin^2\theta (d\phi^2)^2).
\end{equation}
It is important to note that many references work with the angle $\hat\phi = 2\phi^1$ with period $4\pi$. A computation yields the harmonic map
\begin{equation}
\psi_0^1 =\frac{2\sqrt{3} \beta}{n\alpha} \;, \qquad \psi_0^2 = \frac{\sqrt{3}\beta \cos\theta}{\alpha},\qquad  \chi_0 = - \frac{4\sqrt{3} \alpha}{n} \cos\theta,
\end{equation}
\begin{equation}
 \zeta_0^1 =  \frac{16\beta \cos\theta}{n^2},\qquad \zeta_0^2 = -\frac{4\beta \sin^2\theta}{n}.
\end{equation}
Thus the angular momenta and charge, as defined in terms of potentials given above, are
\begin{equation}
	\mathcal{J}_{1} = -\frac{4\pi \beta}{n^2}, \qquad  \mathcal{J}_{2} = 0,\qquad Q = \frac{2\sqrt{3}\pi \alpha}{n}.
\end{equation}
Observe that the angular momentum associated to $\partial_{\phi^2}$ vanishes. However,  the angular momenta defined on the horizon \emph{need not} equal the ones computed at spatial infinity because the gauge fields carry angular momenta in the black hole exterior. This is related to the fact the Maxwell equation has a `source' term so that there are additional volume contributions when comparing integrals over the horizon and the boundary $S^3$ at infinity.  Indeed, the black lens spacetimes must have two non-vanishing angular momenta as measured with respect to observers at infinity.

The area formula is
\begin{equation}
	A_e = \frac{32\pi^2}{n}\sqrt{2\alpha^3 - \beta^2}  =8\pi \sqrt{\frac{4nQ^3}{3\sqrt{3} \pi} - n^2\mathcal{J}_1^2}.
\end{equation}
For $n=1$ this is the well-known formula for the area of the BMPV black hole \cite{Breckenridge:1996is}.

Now we compute the asymptotics of the harmonic data. Observe that the Killing vectors
$\eta_{(1)}=\frac{n}{2}\partial_{\phi^1}-\partial_{\phi^2}$ and $\eta_{(2)}=\frac{n}{2}\partial_{\phi^1}+\partial_{\phi^2}$ vanish at $\theta=0$ and $\theta=\pi$, respectively. Clearly the direction vectors are not same as in Section 6. However, according to \cite{Hollands2012}, there exists a matrix $A\in SL(2,\mathbb{Z})$ such that
\begin{equation}
\hat{\eta}_i=A^j_i\eta_{(j)},\qquad i,j=1,2,
\end{equation}
and $\hat{\eta}_{(i)}\hat{a}^i_+$ and $\hat{\eta}_{(i)}\hat{a}^i_-$ with $\hat{a}_+=(1,0)$ and $\hat{a}_-=(1,n)$ vanish at $\theta=0$ and $\theta=\pi$, respectively. In other words, if we select the functions $\hat{\phi}^i$ such that $\mathfrak{L}_{\hat{\eta}_{(i)}}\hat{\phi}^i=1$, we have
\begin{equation}
\eta_{(1)}=\frac{n}{2}\partial_{\phi^1}-\partial_{\phi^2}=\partial_{\hat{\phi}^1}=\hat{\eta}_i\hat{a}^i_+,\qquad \eta_{(2)}=\frac{n}{2}\partial_{\phi^1}-\partial_{\phi^2}=\partial_{\hat{\phi}^1}
+n\partial_{\hat{\phi}^2}=\hat{\eta}_i\hat{a}^i_-.
\end{equation}
The transformation from $(\phi^1,\phi^2)$ to $(\hat{\phi}^1,\hat{\phi}^2)$ is
\begin{equation}
\begin{pmatrix}
\phi^1  \\
\phi^2
\end{pmatrix}=C\begin{pmatrix}
\hat{\phi}^1   \\
\hat{\phi}^2
\end{pmatrix},\qquad C=\begin{pmatrix}
\frac{n}{2}& 0  \\
-1 & \frac{2}{n}
\end{pmatrix}.
\end{equation}
Then there exists a matrix $B=\begin{pmatrix}
1& 1  \\
0 & n
\end{pmatrix}$
such that $(\hat{\phi}^1,\hat{\phi}^2)^T=B(\bar{\phi}^1,\bar{\phi}^2)^T$. Therefore we obtain
\begin{equation}
\begin{pmatrix}
\phi^1  \\
\phi^2
\end{pmatrix}=Z\begin{pmatrix}
\bar{\phi}^1   \\
\bar{\phi}^2
\end{pmatrix},\qquad Z=CB=\begin{pmatrix}
\frac{n}{2}& \frac{n}{2}  \\
-1 & 1
\end{pmatrix},
\end{equation}
and $\bar{\lambda}=Z^T\lambda Z$. The metric functions are then
\begin{equation}
U_0=\frac{1}{4}\log\left(\frac{\det\lambda}{n^2\sin^2\theta}\right),\qquad W_0=\sinh^{-1}\left(\frac{\frac{n^2}{4}\lambda_{11}-\lambda_{22}}{e^{2U_0}\sin\theta}\right),
\end{equation}
\begin{equation}
V_0=\frac{1}{4}\log\left(\frac{\cos^2\frac{\theta}{2}
\left(\frac{n^2}{4}\lambda_{11}-n\lambda_{12}
+\lambda_{22}\right)}{\sin^2\frac{\theta}{2}
\left(\frac{n^2}{4}\lambda_{11}+n\lambda_{12}+\lambda_{22}\right)}\right),
\end{equation}
\begin{equation}
\begin{pmatrix}
\bar{\psi}^1_0   \\
\bar{\psi}^2_0
\end{pmatrix}=Z^T\begin{pmatrix}
{\psi}^1   \\
{\psi}^2
\end{pmatrix}=\begin{pmatrix}
\frac{n}{2}{\psi}^1-{\psi}^2   \\
\frac{n}{2}{\psi}^1+{\psi}^2
\end{pmatrix}.
\end{equation}
The asymptotics of the harmonic map as $\theta\to 0,\pi$ are as follows:
\begin{equation}
V_0=\frac{1}{2}\log\left(\frac{2\alpha^3}{j^2}\right)+O(\sin^2\theta),\qquad \partial_{\theta} V_0=O(\sin\theta),
\end{equation}
\begin{equation}
U_0=\frac{1}{4}\log\left(\frac{8j^2}{\alpha n^4}\right),\qquad \partial_{\theta} U_0=0,\qquad
W_0=O(\sin\theta),\qquad\partial_{\theta} W_0=O(1),
\end{equation}
\begin{equation}
\bar{\psi}^1_0=O(\sin^2\frac{\theta}{2}),\quad {\partial_{\theta}\bar{\psi}^1_0}=O(\sin\theta),\quad \bar{\psi}^2_0=O(\cos^2\frac{\theta}{2}),\quad {\partial_{\theta}\bar{\psi}^2_0}=O(\sin\theta),
\end{equation}
\begin{equation}
\zeta^1_0=-\frac{2\mathcal{J}_1}{\pi}+O(\sin^2\frac{\theta}{2})\quad\text{ as $\theta\to 0$},\qquad \zeta^1_0=\frac{2\mathcal{J}_1}{\pi}+O(\cos^2\frac{\theta}{2})\quad\text{ as $\theta\to \pi$},
\end{equation}
\begin{equation}
\chi_0=\frac{2Q}{\pi}+O(\sin^2\frac{\theta}{2})\quad\text{ as $\theta\to 0$},\qquad \chi_0=-\frac{2Q}{\pi}+O(\cos^2\frac{\theta}{2})\quad\text{ as $\theta\to \pi$},
\end{equation}
\begin{equation}
\partial_{\theta}\chi_0=O(\sin\theta),\qquad \zeta^2_0=O(\sin^2\theta),\qquad
 {\partial_{\theta}\zeta^i_0}=O(\sin\theta).
\end{equation}

\bibliographystyle{abbrv}
	\bibliography{masterfileArea}

\end{document}